\documentclass{scrartcl}
 
\usepackage{amsmath,amsthm}
\usepackage{amsfonts}
\usepackage{graphicx}
\usepackage{epstopdf}
\usepackage{amssymb}
\usepackage{bm}
\usepackage{algorithmic}
\usepackage{amsopn}
\usepackage{hyperref}
\usepackage{xr}
\usepackage[numbers]{natbib}

\newtheorem{theorem}{Theorem}
\newtheorem{proposition}[theorem]{Proposition}
\numberwithin{theorem}{section}

\newcommand{\TheTitle}{Probabilistic Numerical Methods for Partial Differential Equations and Bayesian Inverse Problems\thanks{Supplementary material for this work is available at \url{http://joncockayne.com/papers/pmm/supplement}. 
  TJS is supported by the Free University of Berlin within the Excellence Initiative of the German Research Foundation (DFG).
MG was supported by EPSRC [EP/J016934/1, EP/K034154/1], an EPSRC Established Career Fellowship, the EU grant [EU/259348] and a Royal Society Wolfson Research Merit Award.
CJO and MG were supported by the Programme on Data-Centric Engineering at the Alan Turing Institute.
This material was based upon work partially supported by the National Science Foundation under Grant DMS-1127914 to the Statistical and Applied Mathematical Sciences Institute. Any opinions, findings, and conclusions or recommendations expressed in this material are those of the author(s) and do not necessarily reflect the views of the National Science Foundation.
}} 
 
\newcommand{\TheAuthors}{Jon Cockayne, Chris Oates, T.\ J.\ Sullivan and Mark Girolami}

\newcommand*{\email}[1]{\href{mailto:#1}{\nolinkurl{#1}} } 

\title{{\TheTitle}}

\author{
  Jon Cockayne\thanks{University of Warwick
    (\email{j.cockayne@warwick.ac.uk}).}
  \and
  Chris J.\ Oates\thanks{Newcastle University and Alan Turing Institute
    (\email{chris.oates@ncl.ac.uk}).}
  \and
  T.\ J.\ Sullivan\thanks{Free University of Berlin and Zuse Institute Berlin
    (\email{sullivan@zib.de}).}
  \and
  Mark Girolami\thanks{Imperial College and Alan Turing Institute
    (\email{m.girolami@imperial.ac.uk}).}
}

\DeclareMathOperator{\diag}{diag}
\DeclareMathOperator*{\argmin}{arg\,min}

\DeclareMathOperator{\trace}{Tr}

\newcommand{\COMMENT}[1]{}

\newcommand{\set}[1]{\left\{#1\right\}}
\newcommand{\norm}[1]{\left\| #1 \right\|}
\newcommand{\inner}[1]{\left\langle #1 \right\rangle}
\newcommand{\abs}[1]{\left| #1 \right|}
\newcommand*{\quark}{\setbox0\hbox{$x$}\hbox to\wd0{\hss$\cdot$\hss}}
\newcommand*{\reals}{\mathbb{R}}

\newcommand*{\laplacian}{\nabla^2}

\newcommand*{\wrt}{\textup{d}}
\newcommand\numberthis{\addtocounter{equation}{1}\tag{\theequation}}

\newcommand{\pnparam}{\zeta}
\newcommand{\pnspace}{\mathcal{Z}}

\ifpdf
\hypersetup{
  pdftitle={\TheTitle},
  pdfauthor={\TheAuthors}
} 
\fi

\usepackage{color}
\usepackage{graphicx}
\usepackage{enumerate}
\usepackage[toc,page]{appendix}

\usepackage{subcaption}
\usepackage{tikz}
\usetikzlibrary{shapes,backgrounds}

\begin{document}
\maketitle
\begin{abstract}
This paper develops a probabilistic numerical method for solution of partial differential equations (PDEs) and studies application of that method to PDE-constrained inverse problems.
This approach enables the solution of challenging inverse problems whilst accounting, in a statistically principled way, for the impact of discretisation error due to numerical solution of the PDE.
In particular, the approach confers robustness to failure of the numerical PDE solver, with statistical inferences driven to be more conservative in the presence of substantial discretisation error.
Going further, the problem of choosing a PDE solver is cast as a problem in the Bayesian design of experiments, where the aim is to minimise the impact of solver error on statistical inferences; here the challenge of non-linear PDEs is also considered.
The method is applied to parameter inference problems in which discretisation error in non-negligible and must be accounted for in order to reach conclusions that are statistically valid.
\end{abstract}
\COMMENT{
}


\section{Introduction}

Differential equations provide a natural language in which to describe many scientific phenomena of interest.
In particular, partial differential equations (PDEs) are widely used to describe phenomena such as heat transfer, electrical conductivity and temporal processes on continuous domains.
Models based on PDEs can involve finite- or infinite-dimensional physical parameters, such as the conductivity field of a heterogeneous media, whose precise values are unknown yet must be specified before the model can be used.
In such situations, statistical methods can be used to \emph{estimate} these parameters.
These statistical methods operate on the basis of data, together with statistical assumptions that describe how the data relate to the posited PDE model.
The applied mathematics literature refers to the task of estimating unknown parameters on the basis of data as an \emph{inverse problem} \cite{Kaipio2006, Stuart2010}.

Several important challenges are raised by the increasing sophistication of mathematical models built with PDEs.
The principal challenge, in general, is that standard methods to solve an inverse problem require that the PDE be solved for numerous candidate values of the unknown parameters.
These solutions are then compared to data to determine which candidate parameters are most plausible, according to the specifics of the statistical method.
For generic PDEs a closed-form solution does not exist, so a discretisation of the continuous equations must be used to approximate the solution.
Theoretical results from numerical analysis typically bound the approximation error as a function of the discretisation parameters, which justifies this approach, but a reduction in approximation error comes at an increased computational cost.
This raises the following question:
is it possible to compute with a coarse discretisation of the PDE and yet produce a meaningful solution to the inverse problem?

This paper gives an affirmative answer by pursuing the development of a \emph{probabilistic numerical method} (PNM, \cite{Hennig2015a}) for the solution of a generic PDE.
The output of a PNM is a probability distribution over the solution space of the PDE, where stochasticity is used as a device to quantify epistemic uncertainty resulting from the discretisation.
Once built, the PNM can be used to solve the PDE up to a quantified degree of uncertainty.
However, unlike classical error estimators for numerical solution of PDEs \citep{Iserles2009}, the nature of this quantification is \emph{statistical}.
This last point is crucial, as it enables uncertainty to be propagated through all subsequent statistical computations and, ultimately, used to inform conclusions that are drawn on the unknown parameters of interest.

This research showcases the emerging field of \emph{probabilistic numerics}, wherein numerical methods are studied from a statistical perspective, which allows them to be assessed, compared and even \emph{designed} according to established statistical criteria.
Diaconis \cite{Diaconis1988} traced the philosophical foundations of probabilistic numerics back to Poincar\'{e}, alongside other landmark and foundational papers \cite{Hull1966, Kadane1985}. 
To date, research in this field includes solvers for linear systems \citep{Bartels:2016eh, Hennig2015}, differential equations \citep{Barber2014, Calderhead2009, Chkrebtii:2013ux, Conrad2015, Dondelinger2013, Kersting2016, Macdonald2015, Owhadi2015, Schober2014, Schober:2016uh} and integrals \citep{Briol2016, OHagan1991}. 
PNMs for the solution of PDEs initially concerned probabilistic models for rounding error in floating point computation \citep{Hull1966}, a topic which is still of some interest \citep{Hairer2008, Mosbach2005}. 
The focus of \cite{Skilling1992} was instead to model discretisation error for ordinary differential equations, similar to our present work on PDEs. The work of \cite{Conrad2015} constructed a PNM by introduction of random perturbations into a finite element basis used in the Galerkin method.
The authors also showed how this uncertainty could be taken into account in the context of inverse problems.
The methods proposed and studied in this work are markedly different from the above, as described next.

\subsection{Contribution}

The aim of this paper is to develop a PNM to solve the strong formulation of a PDE, as opposed to the weak (or \emph{variational}) form considered in \cite{Conrad2015}.
Our approach is modelled on recent work in Bayesian methods \citep{Stuart2010} and begins with a \emph{prior} distribution over the solution space of the PDE.
This prior is then restricted to a subset of the solution space on the basis of information about the true solution of the PDE.
This can be thought of as imposing the governing equations of the PDE at a finite number of locations in the domain of interest, rather than at infinitely many such locations.
The choice of locations at which to impose the equations constitutes the discretisation of the PDE.
The resultant restriction is called a \emph{posterior} distribution and this constitutes the output of our PNM.
The associated stochasticity is used as an abstract device to represent epistemic uncertainty due to discretisation of the original PDE.

The specific contributions of this work are as follows:
\begin{itemize}
\item Presentation of the novel PNM, which we term a \emph{probabilistic meshless method} (PMM).
\item Theoretical analysis of the PMM output; how this contracts to the exact solution of the PDE in the limit of infinite computational effort.
\item Discussion of prior selection for PDE models, with theoretical results for linear PDEs.
\item Theoretical justification for the use of the PMM in PDE-constrained inverse problems.
Estimates for parameters in the PDE are demonstrated to be robust to solver failure, as a consequence of the full quantification of discretisation error being performed.
\item Extension to a class of non-linear PDEs, where uniqueness of solution is not guaranteed.
\item Empirical results based on the use of the PMM in the context of inverse problems that occur in electro-impedance tomography and in a prototypical non-linear PDE.
\end{itemize}

These contributions suggest several areas for future theoretical and methodological development.
Indeed, the development of the PMM serves to highlight the pertinent statistical considerations that inform the development of PNMs in general.
The research landscape in which our paper exists is described next.

\subsection{Related Work}

The method described herein has much in common with symmetric collocation \cite{Fasshauer1999}.
This is an example of a \emph{meshless} method that also enforces the governing equations at a finite number of locations.
Meshless methods for PDEs are characterised by non-reliance on the construction of a mesh over the domain of interest \citep{Fasshauer1996, Franke1998, Hon2008}.
This is an attractive property when the domain itself is time-evolving, as it circumvents the need to re-compute a mesh or grid.
There is empirical evidence in support of meshless methods in the solution of PDEs in situations with strong boundary effects, or when the domain itself is time-dependent and the evolutionary rates are large \citep{Li2010, Li2012}.
However, compared to other numerical methods for solving PDEs, such as finite element methods (FEM) and finite difference methods (FDM), theoretical analysis of collocation methods has been limited, with main contributions including \cite{Behrens2002, Cialenco2012, Lorentz2003, Wendland2004}, and \cite{Chi2013}.
In particular there has been little investigation of the suitability of meshless methods in the inverse problem context.
This PNM developed herein is related to symmetric collocation, due to a connection between collocation and stochastic PDEs that casts the latter as a stochastic relaxation of the former \citep{Cialenco2012}.

The solution of operator equations has been formulated as a statistical problem in several papers (e.g.~\cite{Cialenco2012, Fasshauer2013, Graepel2003, Raissi2017a, Raissi2017, Saerkkae2011, Skilling1992}).
The present approach is close to these papers in spirit, but the motivation of this work is to make valid inferences for the inverse problem, which these papers do not consider.
Recent work by \cite{Barber2014} considered the inverse problem and proposed a method to obtain valid statistical inferences based on a probabilistic model for solver error.
A shortcoming of \cite{Barber2014} is an absence of theoretical analysis, as well as the restriction of attention to finite-dimensional parameters; our work addresses these points.

In the context of an inverse problem, data are related to the PDE model of interest through a statistical \emph{likelihood}.
For PDEs, the use of approximate likelihoods to reduce computational cost in solution of the inverse problem has been widely explored.
Several approaches start by building an approximate solution to the PDE and penalising deviation between derivatives of this approximate solution and the derivative values that are required; see \cite{Campbell2007, Dattner2013, Heinonen2014, Ramsay2007} and the references therein.

Closer in spirit to the approach pursued below, \cite{Marzouk2009} used polynomial chaos expansions to approximate the likelihood, treating the mapping from parameters to data as a black-box.
A similar approach was proposed in \cite{Webster1996} and \cite{Ma2009}.
Both \cite{Marzouk2009} and \cite{Ma2009} apply their likelihood approximations to solve inverse problems and \cite{Marzouk2009} establishes that the approximate and exact posterior coincide in an appropriate limit.
However, these papers do not take into account the approximation error when making inferences, meaning that careful control of error in the forward problem is still required to avoid the problems of bias and over-confidence exposed in \cite{Conrad2015}.
Recent work of \cite{Stuart2016} constructed Gaussian process emulators and integrated their associated uncertainties into solution of an inverse problem; our work differs in that it exploits the specific form of the PDE model rather than treating the PDE as a black-box.

The solution of differential equations has been a recent focus of PNM development. 
For ordinary differential equations, \cite{Conrad2015} constructed a probabilistic method for modelling discretisation error in numerical integrators, while \cite{Schober2014} revealed the underlying uncertainty model that is implied by Runge--Kutta methods and \cite{Schober:2016uh} provided a connection between these solvers and Nordsieck methods. 
See also the recent contribution by \cite{Kersting2016}.
A similar approach is followed in \cite{Chkrebtii:2013ux}, yielding a nonparametric posterior rather than a Gaussian approximation.
For PDEs, \cite{Kaipio2007} fitted Gaussian models for PDE errors and \cite{Conrad2015} proposed PNM interpretations of FEM (see also \cite{Arnold2013}).
Related work by \cite{Capistran2016} analysed the impact of discretisation error using Bayes factors.

\subsection{Outline}

The paper proceeds as follows:
Section~\ref{context and motivation} establishes the set-up and notation, while Section~\ref{sec:forward} outlines the proposed method.
Sections~\ref{sec:forward_error} and \ref{sec:backward_error} provide error analysis for, respectively, the forward and inverse problems.
Computational considerations are discussed in Section~\ref{sec:compute}.
Section~\ref{sec:applications} provides empirical results on the proposed approach, with discussion reserved for Section~\ref{sec:discussion}.

\section{Background} \label{context and motivation}

In this section we first formulate the notion of an inverse problem, then expand on the statistical motivation for valid inferences in the presence of discretisation error.

\subsection{Inverse Problems} \label{sec:setup}

The physical models considered here are defined via operator equations.
An inverse problem arises when some of these operators depend upon unknown parameters that must be inferred. 
This enables predictions to be obtained under the model, or provides insight into the physical system of interest.
The Bayesian approach treats unknown parameters as random variables; below we set up a mathematical framework that makes both the parameter dependence and the randomness explicit.

\subsubsection{Set-up and Notation} \label{sec:notation}

Consider a compact domain $D \subset \mathbb{R}^d$ with Lipschitz boundary $\partial D$.
Let $(\Omega,\mathcal{F}_\Omega,\mathbb{P}_\Omega)$ be a probability space and consider measurable operators $\mathcal{A} \colon \Omega \times H(D) \to H_{\mathcal{A}}(D)$ and $\mathcal{B} \colon \Omega \times H(D) \to H_{\mathcal{B}}(\partial D)$ among Hilbert spaces of functions $H(D)$, $H_{\mathcal{A}}(D)$, and $H_{\mathcal{B}}(\partial D)$.
Let $\bar{\mathcal{A}}$ denote the adjoint\footnote{In particular, when later in this paper the operator $\mathcal{A}$ is applied to a kernel function of two arguments, $\mathcal{A}$ refers to action on the first argument, while the adjoint $\bar{\mathcal{A}}$ refers to action on the second argument.} of $\mathcal{A}$.

Consider the stochastic solution $u(\quark,\omega) \in H(D)$, $\omega \in \Omega$, of operator equations of the form
\begin{alignat}{2}
	\mathcal{A}[\omega] u(\bm{x},\omega) &= g(\bm{x}) &\quad& \bm{x} \in D \nonumber \\
	\mathcal{B}[\omega] u(\bm{x},\omega) &= b(\bm{x}) &\quad& \bm{x} \in \partial D \label{eq:syst1}
\end{alignat}
where $g \in H_{\mathcal{A}}(D)$ and $b \in H_{\mathcal{B}}(\partial D)$.
The notation $\mathcal{A}[\omega]$ and $\mathcal{B}[\omega]$, $\omega \in \Omega$ is used to emphasise the random nature of the operators $\mathcal{A}$ and $\mathcal{B}$, perhaps as a result of the dependence on unknown parameters in a Bayesian approach, as will be described in the following section.
For concreteness, one can associate $\mathcal{A}$ with a PDE to be solved and $\mathcal{B}$ with any initial or boundary conditions. 
Similarly, $g \in H_{\mathcal{A}}(D)$ and $b \in H_{\mathcal{B}}(\partial D)$ can be considered as forcing and boundary terms for the PDE.

For notational simplicity we will generally restrict attention to systems with two operators as in Eq.~\eqref{eq:syst1}, however it is trivial to extend the algorithm of this paper to systems of more than two operators, each potentially restricted to subsets of $D$.

An inverse problem is one in which inferences are to be made for $\omega$, on the basis of possibly noisy observations of the underlying solution $u(\quark,\omega^\dagger)$, or derived quantities thereof, where $\omega^\dagger \in \Omega$ is the ``true'' value of $\omega$.
Typically both $\mathcal{A}$ and $\mathcal{B}$ depend on $\omega$ through some $\theta(\omega)$ that is of physical interest, where $\theta \colon \Omega \to \Theta$ is a measurable function mapping into a separable Banach space $\Theta$ with the Borel $\sigma$-algebra $\mathfrak{B}(\Theta)$.
The true value of $\theta$, denoted $\theta^\dagger = \theta(\omega^\dagger)$, is the object of statistical interest.

\subsubsection{The Bayesian Approach} \label{sec:bayesian_inverse_prob}

The Bayesian approach endows $(\Theta,\mathfrak{B}(\Theta))$ with a prior distribution $\Pi_{\theta}$.
The prior is updated on the basis of data, which in this paper refers to a random variable $\bm{y}$ with distribution $\Pi_{\bm{y}}$, defined on $(\mathcal{Y},\mathfrak{B}(\mathcal{Y}))$, where $\mathcal{Y} \subset \reals^n$ is equipped with the Borel $\sigma$-algebra $\mathfrak{B}(\mathcal{Y})$. 
The conditional density
\begin{equation*}
	\pi(\bm{y}|\theta) \propto \exp( -\Phi(\bm{y},\theta)), \label{finite data}
\end{equation*}
is called the \emph{likelihood};
$\Phi \colon \mathcal{Y} \times \Theta \to \mathbb{R}$ is a measurable function referred to variously as the \emph{potential}, or the \emph{negative log-likelihood} of $\bm{y}$ conditioned upon $\theta$.
An infinite-dimensional analogue of Bayes' theorem \cite[Theorem 1.1]{Dashti2014} implies the existence of a posterior distribution $\Pi_\theta^{\bm{y}}$ on $(\Theta, \mathfrak{B}(\Theta))$ that is absolutely continuous with respect to $\Pi_\theta$, with Radon--Nikod\'{y}m derivative
\begin{equation}
\frac{\wrt \Pi_\theta^{\bm{y}}}{\wrt \Pi_\theta}(\theta) = \frac{1}{Z} \exp(- \Phi(\bm{y} , \theta)) , \; \; \; Z = \int_\Theta \exp( - \Phi(\bm{y},\theta)) \Pi_\theta(\wrt \theta) \label{eq:bayes_first}
\end{equation}
whenever $Z > 0$.
In colloquial use in the applied mathematical context, the \emph{Bayesian inverse problem} (BIP) entails numerical computation of the posterior distribution $\Pi_\theta^{\bm{y}}$, or derived quantities thereof.

\subsection{Statistical Motivation: Valid Inference at Lower Cost} \label{sec:valid_inference}

We now expand on the motivation for developing a more expressive quantification of numerical error in the forward problem, within the context of the statistical inverse problem.
Consider for example the Gaussian measurement error model, with the potential
\begin{equation*}
\Phi(\bm{y},\theta) = \frac{1}{2} \| \bm{y} - \mathcal{G}(\theta) \|_\Gamma^2,
\end{equation*}
where $\mathcal{G} \colon \Theta \to \mathcal{Y}$ is a parameter-to-observable map and data $\bm{y} \in \mathcal{Y}$ are observations that have been collected.
Here the symmetric positive semi-definite $n \times n$ matrix $\Gamma$ defines an appropriate scaling of the residual vector and hence also a Cameron--Martin space \cite[Section~7.3]{Dashti2014}.

Typically an analytic representation for $\mathcal{G}(\theta)$ is unavailable, so that a numerical solver is used to obtain an approximation $\hat{\mathcal{G}}(\theta)$.
Inference then proceeds based on the approximate potential
\begin{equation*}
\hat{\Phi}(\bm{y},\theta) = \frac{1}{2} \| \bm{y} - \hat{\mathcal{G}}(\theta) \|_\Gamma^2
\end{equation*}
in place of the true potential $\Phi$.
For PDEs, the difference between $\hat{\Phi}$ and $\Phi$ can typically be driven to negligible values by running numerical methods on a detailed discretisation of the PDE \citep{Schwab2012}.
These techniques combine to produce statistically valid inferences on the parameter $\theta$ \citep{Stuart2010}.
However, the requirement to drive error to negligible values can have a high computational cost.

Instead, we propose here a novel approach based on probabilistic solution of the forward problem, in which error in the discretisation of the forward problem is captured \emph{statistically} and accounted for in inferences made in the BIP.
Capturing this error can permit an overall reduction in computation in some situations, by allowing use of a coarser discretisation while still yielding statistically valid inferences.

\section{Methods} \label{sec:forward}

In this section a probabilistic meshless method is formally defined.
The starting point is radial basis function collocation, as studied by \cite{Fasshauer1996} and more recently by \cite{Owhadi2015}. 
Initially it is assumed that the operators $\mathcal{A}$ and $\mathcal{B}$ are linear; this will be relaxed in Section~\ref{sec:nonlinear}. 

\subsection{Probability Measures for Solutions of PDEs}

Let $(\pnspace,\mathcal{F}_\pnspace,\mathbb{P}_\pnspace)$ be a second probability space and consider a measurable function $g \colon D \times \pnspace \to \mathbb{R}$ such that, for each $\pnparam \in \pnspace$, $g(\quark,\pnparam) \in H_\mathcal{A}(D)$. Similarly consider a measurable function $b \colon \partial D \times \pnspace \to \mathbb{R}$ such that for each $\pnparam \in \pnspace$, $b(\quark,\pnparam) \in H_\mathcal{B}(\partial D)$.
The mathematical object studied in this section is the ``doubly stochastic'' solution $u(\quark,\omega,\pnparam) \in H(D)$, $\omega \in \Omega$, $\pnparam \in \pnspace$, of operator equations of the form
\begin{alignat}{2}
	\mathcal{A}[\omega] u(\bm{x},\omega,\pnparam) &= g(\bm{x},\pnparam) &\quad& \bm{x} \in D \nonumber \\
	\mathcal{B}[\omega] u(\bm{x},\omega,\pnparam) &= b(\bm{x},\pnparam) &\quad& \bm{x} \in \partial D. \label{eq:syst2}
\end{alignat}
The system in Eq.~\eqref{eq:syst2} is a stochastic relaxation of the original inverse problem in Eq.~\eqref{eq:syst1}, in the sense that the \emph{deterministic} forcing terms $g$ and $b$ are, for the purposes of exposition, formally considered as \emph{random fields}. 
This construction is justified in \cite{Owhadi2015} as a reflection of the epistemic uncertainty, from the perspective of the numerical solver, about the value of the forcing at locations where it has not been evaluated.

The doubly stochastic solution $u$ exists as a random variable that takes values $(\mathbb{P}_\Omega,\mathbb{P}_{\pnspace})$-almost surely in an appropriate function space; this will be made precise in Section~\ref{sec:existence_of_RV}.
The next section focuses on the effect of this randomisation and how it connects, at a deep level, to collocation methods.

\subsubsection{Probabilistic Meshless Method}
\label{sec:pmm}

Fasshauer \cite{Fasshauer1996, Fasshauer2011} observed that collocation methods based upon radial basis functions implicitly posit a reproducing kernel Hilbert space (RKHS) for the solution $u$, with a kernel $k$.
Here, similarly to \cite{Cialenco2012}, we extend this viewpoint by positing a Gaussian process prior for $u$ with covariance function $k$. 
That is, for fixed $\omega$, the map $\pnparam \mapsto u(\cdot, \omega, \pnparam)$ is a stochastic process whose finite dimensional marginals $[u(\bm{x}_1, \omega, \quark),\dots,u(\bm{x}_n, \omega, \quark)]$ have a Gaussian distribution for any $\set{\bm{x}_1, \dots, \bm{x}_n} \subset D$.
The mean vector and covariance matrix of these marginals are characterised by the mean function $m \colon D \to \mathbb{R}$ and covariance function $k \colon D \times D \to \mathbb{R}$ of the Gaussian process, which completely characterise the distribution $\Pi_u$. 
Throughout, the prior mean function is assumed to be zero; this assumption can be trivially relaxed.
Choice of the kernel $k$ will be discussed in Sections~\ref{sec:natural_space} and \ref{sec:posited_kernel}. 
For the remainder of Section \ref{sec:forward} we leave all dependence on $\omega$ and $\pnparam$ implicit.

Now, a posterior measure is constructed that represents epistemic uncertainty over the solution $u$ after expending a finite amount of computational effort.
To accomplish this, the prior measure $\Pi_u$ is conditioned on $m_{\mathcal{A}} \in \mathbb{N}$ evaluations of the forcing function at distinct locations $X_0^{\mathcal{A}} = \{\bm{x}_{0,j}^{\mathcal{A}}\}_{j=1}^{m_{\mathcal{A}}} \subset D$, and $m_{\mathcal{B}} \in \mathbb{N}$ evaluations of the boundary function at locations $X_0^{\mathcal{B}} = \set{\bm{x}_{0,j}^{\mathcal{B}}}_{j=1}^{m_{\mathcal{B}}} \subset \partial D$. These are referred to as the \emph{design} points.
The sought-for solution $u$ and the evaluations are related through the interpolation equations
\begin{align*}
	\mathcal{A} u(\bm{x}_{0,j}^{\mathcal{A}}) &= g(\bm{x}_{0,j}^{\mathcal{A}}), \; \; \; j = 1,\dots,m_{\mathcal{A}} \\
	\mathcal{B} u(\bm{x}_{0,j}^{\mathcal{B}}) &= b(\bm{x}_{0,j}^{\mathcal{B}}), \; \; \; j = 1,\dots,m_{\mathcal{B}}. \numberthis \label{eq:obs}
\end{align*}
Write $\bm{g}$ for the $m_{\mathcal{A}} \times 1$ vector with $j$\textsuperscript{th} element $g(\bm{x}_{0,j}^{\mathcal{A}})$ and $\bm{b}$ the respective vector for points in $X_0^{\mathcal{B}}$.
Then the conditional process $u|\bm{g}, \bm{b}$, denoted by $\Pi_u^{\bm{g}, \bm{b}}$, is also Gaussian and is characterised by its finite-dimensional marginals, given in Proposition~\ref{prop:forward_posterior}. To construct this posterior some notation must first be established.

For sets $X = \{\bm{x}_j\}_{j=1}^n$ and $X' = \{\bm{x}_j'\}_{j=1}^{n'}$, denote by $\bm{K}(X,X')$ the $n \times n'$ matrix whose $(i,j)$\textsuperscript{th} element is $k(\bm{x}_i,\bm{x}_j')$.
When $X = X'$ the notation $\bm{K}(X) = \bm{K}(X,X)$ is used.
The $n \times n'$ matrices $\mathcal{A} \bm{K}(X,X')$, $\bar{\mathcal{A}} \bm{K}(X,X')$ and $\mathcal{A} \bar{\mathcal{A}} \bm{K}(X,X')$ have respective $(i,j)$\textsuperscript{th} entries $\mathcal{A} k(\bm{x}_i,\bm{x}_j')$, $\bar{\mathcal{A}} k(\bm{x}_i,\bm{x}_j')$ and $\mathcal{A} \bar{\mathcal{A}} k(\bm{x}_i,\bm{x}_j')$. 
Define
\begin{equation*}
\mathcal{L} := \begin{bmatrix} \mathcal{A} \\ \mathcal{B} \end{bmatrix}, \; \; \; \bar{\mathcal{L}} := \begin{bmatrix} \bar{\mathcal{A}} & \bar{\mathcal{B}} \end{bmatrix} .
\end{equation*}
Introduce the $(m_{\mathcal{A}} + m_{\mathcal{B}}) \times (m_{\mathcal{A}} + m_{\mathcal{B}})$ matrix
\begin{align*}
  \mathcal{L} \bar{\mathcal{L}} \bm{K}(X_0) :=
    \begin{bmatrix}
      \mathcal{A}\bar{\mathcal{A}} \bm{K}(X_0^\mathcal{A}, X_0^\mathcal{A})
        & \mathcal{A}\bar{\mathcal{B}} \bm{K}(X_0^\mathcal{A}, X_0^\mathcal{B}) \\
      \bar{\mathcal{A}}\mathcal{B} \bm{K}(X_0^\mathcal{B}, X_0^\mathcal{A})
        & \mathcal{B}\bar{\mathcal{B}} \bm{K}(X_0^\mathcal{B}, X_0^\mathcal{B})
    \end{bmatrix}
\end{align*}
\noindent
and also the $1 \times (m_{\mathcal{A}} + m_{\mathcal{B}})$ vectors
\begin{align*}
  \bar{\mathcal{L}} \bm{K}(\bm{x},X_0) & :=
    \begin{bmatrix}
      \bar{\mathcal{A}} \bm{K}(\bm{x},X_0^\mathcal{A})
        \\ \bar{\mathcal{B}} \bm{K}(\bm{x},X_0^\mathcal{B})
    \end{bmatrix} &
  \mathcal{L} \bm{K}(\bm{x},X_0) & :=
    \begin{bmatrix}
      \mathcal{A} \bm{K}(\bm{x},X_0^\mathcal{A})
        \\ \mathcal{B} \bm{K}(\bm{x},X_0^\mathcal{B})
    \end{bmatrix} .
\end{align*}

\begin{proposition}[Probabilistic Meshless Method; PMM] \label{prop:forward_posterior}
Let $X = \{\bm{x}_j\}_{j=1}^n \subset D$.
	Denote by $\bm{u}$ the $n \times 1$ vector with $j$\textsuperscript{th} element $u(\bm{x}_j)$.
	Then under $\Pi_u^{\bm{g},\bm{b}}$ we have
	\[
		\bm{u} | \bm{g} , \bm{b} \sim N(\bm{\mu} , \bm{\Sigma})
	\]
	where the posterior mean and covariance are
	\begin{align}
		\bm{\mu} & = \bar{\mathcal{L}} \bm{K}(X,X_0) [\mathcal{L}\bar{\mathcal{L}}\bm{K}(X_0)]^{-1} [\bm{g}^\top  \; \bm{b}^\top ]^\top  \label{eq:full_posterior_mean} \\
		\bm{\Sigma} & = \bm{K}(X) - \bar{\mathcal{L}}\bm{K}(X,X_0) [ \mathcal{L}\bar{\mathcal{L}}\bm{K}(X_0)]^{-1} \mathcal{L}\bm{K}(X_0,X) . \label{eq:full_posterior_cov}
	\end{align}
\end{proposition}

This clarifies what constitutes a ``probabilistic'' numerical method; rather than returning only an approximation to $u$, a probabilistic solver returns a full distribution $\Pi_u^{\bm{g}, \bm{b}}$ where randomness represents uncertainty over the true values of $u$ due to having only evaluated $g$ and $b$ at a finite number of locations.
In Section~\ref{sec:forward_error} it is proven that this statistical quantification of uncertainty is valid.

It is convenient to express the pointwise conditional mean and variance as $\mu(\bm{x})$ and $\sigma(\bm{x})^2$ respectively; i.e.~as defined by Eqns. \eqref{eq:full_posterior_mean} and \eqref{eq:full_posterior_cov} with $X = \{\bm{x}\}$.
Then the expression presented here for $\mu(\bm{x})$ is identical to the meshless method known as \emph{symmetric collocation}, developed by \citep{Fasshauer1999}.
The probabilistic interpretation of symmetric collocation was previously noted in \cite{Cialenco2012}.
Compared to previous literature, the variance term $\sigma^2(\bm{x})$ will play a more central role in this work and will enable formal quantification of numerical error.
While not investigated here, meshless methods can be extended in several directions, including to multi-level methods \citep{Fasshauer1999}.

It remains to discuss the choice of prior measure $\Pi_u$. Two possible choices are presented in the following sections.

\subsubsection{A Natural Prior Measure} \label{sec:natural_space}

The construction described below follows \cite{Owhadi2015}, but operates on functional Hilbert spaces rather than spaces of generalised functions.
Therein the forcing $g$ is formally modelled as a Gaussian stochastic process defined on $H_{\mathcal{A}}(D)$.
The notation $\Pi_g = N(0,\Lambda)$ will be used. 
It will be assumed that there exist fixed linear integro-differential operators $\mathcal{A}_\Lambda$ and $\mathcal{B}_\Lambda$ such that 
\begin{alignat}{2}
	\mathcal{A}_\Lambda g(\bm{x}) &= \xi(\bm{x}) &\quad& \bm{x} \in D \nonumber \\
	\mathcal{B}_\Lambda g(\bm{x}) &= 0 &\quad& \bm{x} \in \partial D \label{eq:laplacian}
\end{alignat}
where $\xi$ is the standard white-noise process.
A common choice for $\mathcal{A}_\Lambda$ is the fractional Laplacian, in which case $\Lambda$ corresponds to a Mat\`{e}rn kernel.
A comprehensive background reference is \cite{Berlinet2011}.

The RKHS corresponding to the kernel $\Lambda$ is denoted $H_\Lambda(D)$.
Denote $\|u\|_2^2 = \int_D u(\bm{x})^2 \wrt \bm{x}$.
It will be assumed that $H_\Lambda(D) \subseteq H_{\mathcal{A}}(D)$.
By construction $H_\Lambda(D)$ contains all functions $g \colon D \to \mathbb{R}$ for which the norm $\|g\|_\Lambda := \|\mathcal{A}_\Lambda g\|_2$ is finite.
The fractional Laplacian choice for $\mathcal{A}_\Lambda$ implies that $H_\Lambda(D)$ is a standard Sobolev space $\mathbb{H}^{\alpha}(D)$ for some order $\alpha$.
An important property of this characterisation is that --- except in the trivial finite-dimensional case --- the Gaussian measure assigns zero mass to the RKHS, i.e.~$\Pi_g[H_\Lambda(D)] = 0$ \citep{Berlinet2011}.
This leads to some additional technical detail in Section~\ref{sec:posited_kernel}.

Next, uncertainty is formally is propagated from the forcing term to the solution of the PDE.
Define the inner product space $( H_{\textup{nat}}(D) , \inner{ \quark , \quark}_{\textup{nat}} )$ by
\begin{align*}
	H_{\textup{nat}}(D) & := \{v \in H(D) \; | \; \mathcal{A}_\Lambda \mathcal{A}v \in L^2(D), \; \mathcal{B}v = 0 \text{ on } \partial D \text{ and } \mathcal{B}_\Lambda \mathcal{A} v = 0 \text{ on } \partial D\} , \\
	\inner{u , v}_{\textup{nat}} & := \int_D [\mathcal{A}_\Lambda \mathcal{A} u(\bm{x})] [\mathcal{A}_\Lambda \mathcal{A} v(\bm{x})]  \wrt\bm{x} .
\end{align*}
Under this definition $\|u\|_{\textup{nat}}^2 := \langle u , u \rangle_{\textup{nat}} = \|g\|_\Lambda^2$.
Proposition~\ref{RKHS} below establishes that $H_\textup{nat}(D)$ is in fact an RKHS for an appropriate choice of kernel. 
Assume non-degeneracy, so that $\| v \|_\textup{nat} = 0$ if and only if $v=0$. Further, assume that the problem is well-posed, meaning that, for any $g \in H_\Lambda(D)$, there exists a unique solution $u \in H_\textup{nat}(D)$ to the system $\mathcal{A}u = g$.
To elicit the reproducing kernel, suppose that we have a Green's function $G$ satisfying
\begin{alignat}{2}
	\mathcal{A} G(\bm{x},\bm{x}') & = \delta(\bm{x} - \bm{x}') &\quad& \bm{x} \in D \nonumber \\
	\mathcal{B} G(\bm{x},\bm{x}') & = 0 &\quad& \bm{x} \in \partial D,
\end{alignat}
and define the \emph{natural kernel} $k_\textup{nat} \colon D \times D \to \mathbb{R}$ by
\begin{equation} \label{eq:k_definition}
	k_\textup{nat}(\bm{x},\bm{x}') := \int_D \int_D G(\bm{x},\bm{z}) G(\bm{x}',\bm{z}') \Lambda(\bm{z},\bm{z}') \wrt\bm{z} \wrt\bm{z}'.
\end{equation}

\begin{proposition} \label{RKHS}
	Assume that $\sup_{\bm{x} \in D} k_\textup{nat}(\bm{x},\bm{x}) < \infty$. Then $H_\textup{nat}(D)$ is a reproducing kernel Hilbert space and $k_\textup{nat}$ is its reproducing kernel.
\end{proposition}
The relationship between Green's functions and kernels is explored in detail by \cite{Fasshauer2011}.
The kernel $k_\textup{nat}$ is indeed natural, in the sense that $H_\Lambda(D)$ is the image under $\mathcal{A}$ of $H_\textup{nat}(D)$.
In the linear case, a realisation of $g$ corresponds to a unique realisation of $u$ and the randomness $\omega \in \Omega$ implies a reference measure $\Pi_u$ over $H(D)$.
Indeed, we have the following:
\begin{proposition} \label{noise model}
	$\Pi_g$ is a mean-zero Gaussian process with covariance function $\Lambda$ if and only if $\Pi_u$ is a mean-zero Gaussian process with covariance function $k_\textup{nat}$.
\end{proposition}

In practice one can specify either the form of $\Lambda$ or the form of $k_\textup{nat}$, since in the linear case each fully determines the other.

Of note is that, when the natural kernel is chosen for the prior covariance in the previous section, the boundary conditions are encoded in the prior. 
As a result, in Proposition~\ref{prop:forward_posterior}, collocation points on the boundary can be omitted which simplifies Eqns \ref{eq:full_posterior_mean} and \ref{eq:full_posterior_cov} somewhat, though the construction is otherwise identical.

\subsubsection{A Practical Prior Measure} \label{sec:posited_kernel}

The presentation above assumes access to a Green's function for the PDE.
In practice Green's functions are not generally available for nontrivial PDE systems. Furthermore the integral in \eqref{eq:k_definition} poses a problem even when the Green's function is accessible. Thus in the general setting an alternative choice of covariance function must be used.

It is often straightforward to elicit a kernel $\tilde{k}$ such that $H_\textup{nat}(D)$ is embedded in $H_{\tilde{k}}(D)$.
A Hilbert space $H$ is said to be (\emph{continuously}) \emph{embedded} in another Hilbert space $H'$ if $H \subseteq H'$ and there is a constant $0 < c < \infty$ such that $\|u\|_{H'} \leq c \|u\|_H$ for all $u \in H$ \citep{Pillai2007}.
In the notation of Section~\ref{sec:notation}, set $H(D) = H_{\tilde{k}}(D)$, an RKHS with reproducing kernel $\tilde{k}$.
In this paper $\tilde{k}$ will often be a kernel whose native space is a Sobolev space, such as a Mat\`{e}rn or Wendland kernel, chosen on a PDE-theoretic basis to ensure that the true solution to the PDE lies in this native space.
The order of this space can be chosen by ``derivative counting'', to reflect the number of (weak) derivatives that $u$ is believed to have based on the maximum differential order of operators in the system.

Due to the aforementioned technicality that $\mathbb{P}_{\pnspace}$ is not supported on the RKHS, the kernel used for the prior measure on $u$ will not be $\tilde{k}$ but instead will be a new kernel $\hat{k}$, whose corresponding Gaussian measure has support on the RKHS $H(D)$.
In particular, it will be required that $H_{\hat{k}}(D)$ is embedded in $H(D)$ and that $\hat{k}$ satisfy the following properties:
\begin{enumerate}[(i)]
\item the measure $\Pi_u = N(0,\hat{k})$ satisfies $\Pi_u[H(D)] = 1$;
\item the set $H_{\hat{k}}(D)$ is dense in the space $(H(D),\|\quark\|_{\tilde{k}})$.
\end{enumerate}
These conditions enable any function $u \in H_k(D) \subseteq H(D)$ to be ``inferred'' from data, under a prior $\Pi_u$.
One choice of kernel that satisfies both (i) and (ii), suggested by \cite[Lemma~2.2]{Cialenco2012}, is the following integral-type kernel $\hat{k}$:

\begin{proposition} \label{prop:integral_type_kernel}
For any  $\tilde{k}$ such that $H_\textup{nat}(D)$ is embedded in $H_{\tilde{k}}(D)$, the kernel 
\begin{equation*}
		\hat{k}(\bm{x}, \bm{x}^\prime) := \int_D \tilde{k}(\bm{x}, \bm{z}) \tilde{k}(\bm{z}, \bm{x}^\prime) \wrt\bm{z}
\end{equation*}
satisfies requirements (i) and (ii) above.
\end{proposition}

\subsubsection{Illustrative Example: Forward Problem} \label{sec:illustrative_forward}

To illustrate these ideas, we examine the above procedure for Poisson's equation in one dimension. Consider the toy system
\begin{alignat*}{2}
	-\laplacian u(x) &= g(x) & \quad \text{for} \; & x \in (0,1) \\
	u(x) &= 0 &\quad \text{for} \; & x \in \{ 0, 1 \} ,
\end{alignat*}
for which the Green's function can be computed:
\begin{align*}
	G(x, x') &= \left\{\begin{array}{rl}
		x(x'-1) & \quad \text{for} \; x > x' \\
		x'(x-1) & \quad \text{for} \; x < x' .
	\end{array}\right.
\end{align*}
Place a Gaussian measure on the forcing term $g$, using the compactly supported polynomial kernel of \cite{Wendland1995}:
\begin{align*}
	\Lambda(x, x') &= \max(1 - \epsilon^{-1} | x-x' |, 0)^2 
\end{align*}
where $\epsilon$ is a parameter which controls the support of the kernel, so that  $\Lambda$ has support wherever $|x-x'| < \epsilon$. 
Samples from the prior $\Pi_g$ will be continuous, but do not have continuous derivatives.
Associating operators in the above system with the abstract formulation, we have $\mathcal{A} := -\laplacian = -\frac{\wrt^2}{\wrt x^2}$, while $\bar{\mathcal{A}} = -\frac{\wrt^2}{\wrt x'^2}$.  
The natural kernel
\begin{align*}
	k_\textup{nat}(x, x') = \int_0^1 \int_0^1 G(x,z) G(x', z') \Lambda(z, z') \wrt z \wrt z'
\end{align*}
is available in closed form since $G$ and $\Lambda$ are each piecewise polynomial. 

Next, a set of design points $\set{x_i}$ must be selected. 
For illustration take $m_{\mathcal{A}} = 39$ function evaluations at evenly spaced points in $(0,1)$.
In Figure~\ref{fig:running_example_natural} the conditional mean is plotted for the above PDE with $g(x) = \sin(2\pi x)$, along with sample paths from the full conditional measure.
The covariance $\Lambda(x,x')$ is assigned a support of $\epsilon = 0.4$. This is contrasted with the closed-form solution $u(x) = (2\pi)^{-2} \sin(2\pi x)$.

Even for this most simple of examples, computation of the natural kernel is challenging.
In practice collocation methods operate using a kernel such as $\hat{k}$ as given in Section~\ref{sec:posited_kernel}, or even by directly positing a kernel.
In Figure~\ref{fig:running_example_integral} the performance of the natural kernel $k_\textup{nat}$ is contrasted with that of $\hat{k}$, computed from a higher-order Wendland covariance function:
\begin{equation*}
	\tilde{k}(x, x') = \max(1- \epsilon^{-1} |x-x'|, 0)^4 \cdot (4\epsilon^{-1} |x-x'| + 1) ,
\end{equation*}
with $\hat{k}$ computed from $\tilde{k}$ as in Proposition~\ref{prop:integral_type_kernel}. This kernel corresponds to the number of derivatives implied by the Laplacian, as $\tilde{k}$ is twice differentiable at the origin.
The design is also augmented with $X_0^{\mathcal{B}} = \{0,1\}$ so that all samples from the conditional measure satisfy the boundary conditions.

Figure~\ref{fig:kernel_convergence} shows convergence of the conditional measures $\Pi_u^{\bm{g}}$ and $\Pi_u^{\bm{g},\bm{b}}$ based on these two kernels as the number of design points is increased. The advantage of using the natural kernel is a reduction in approximation error, though not to the extent of an appreciable change in the convergence rate;
this is to be expected, since both kernels have the same native Sobolev space.
Furthermore, for the natural kernel, the uncertainty in the posterior distribution appears to be more representative of how well the mean function approximates the truth.
A discussion on calibration of the kernel is reserved for Section \ref{sec:calib}.

\begin{figure}
	\centering
	\begin{subfigure}[t]{0.45\textwidth}
		\includegraphics[width=\textwidth]{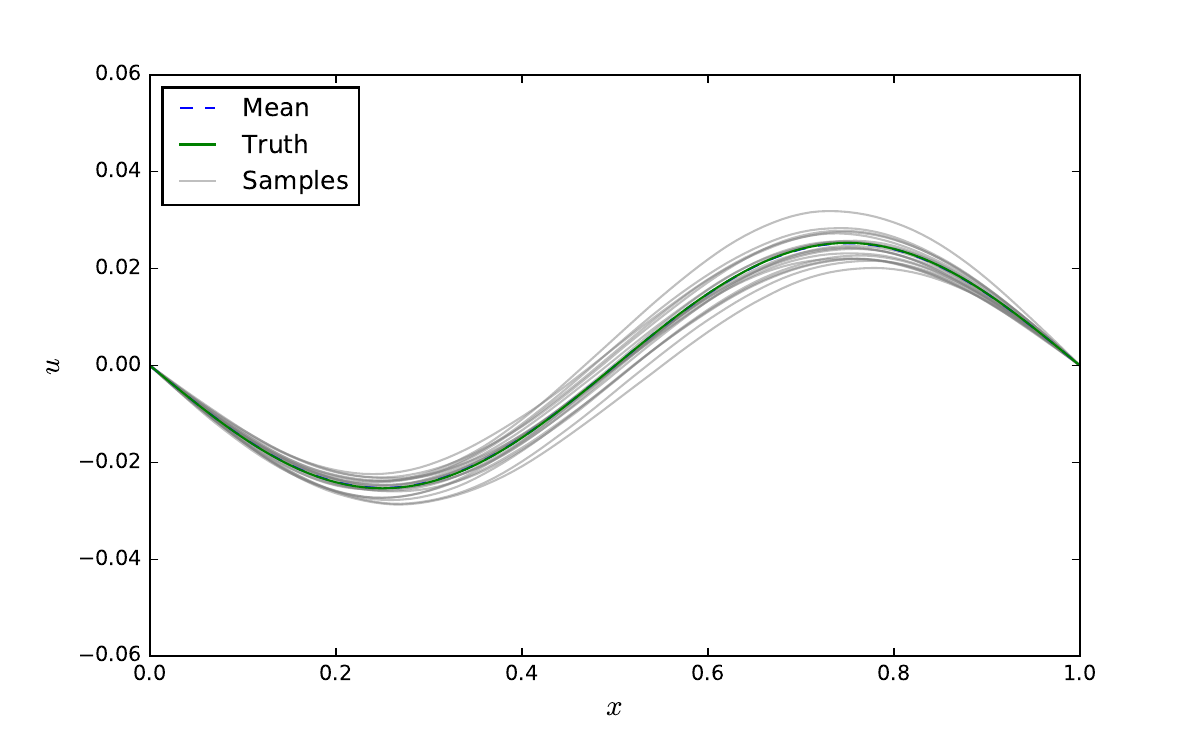}
		\caption{$\Pi_u^{\bm{g}}$, based on $k$} \label{fig:running_example_natural}
	\end{subfigure}
	~
	\begin{subfigure}[t]{0.45\textwidth}
		\includegraphics[width=\textwidth]{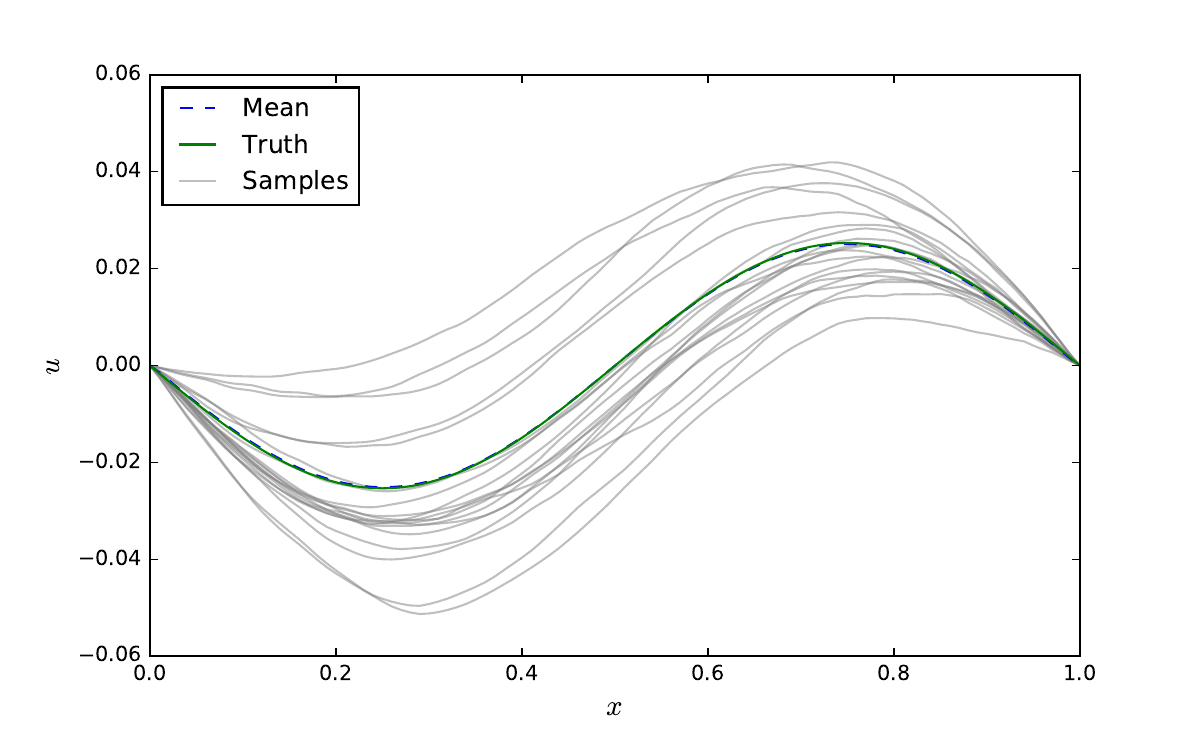}
		\caption{$\Pi_u^{\bm{g},\bm{b}}$, based on $\hat{k}$} \label{fig:running_example_integral}
	\end{subfigure}
	\caption{Probabilistic meshless methods: Comparison of conditional distributions (a) $\Pi_u^{\bm{g}}$ based on the natural kernel $k$ and (b) $\Pi_u^{\bm{g},\bm{b}}$ based on the integrated Wendland kernel $\hat{k}$. 
	In (b) two additional evaluations are performed at $x=0$ and $x=1$ to enforce the boundary conditions.}
	\label{fig:running_example}
\end{figure}

\begin{figure}
	\centering
	\begin{subfigure}[t]{0.45\textwidth}
		\includegraphics[width=\textwidth]{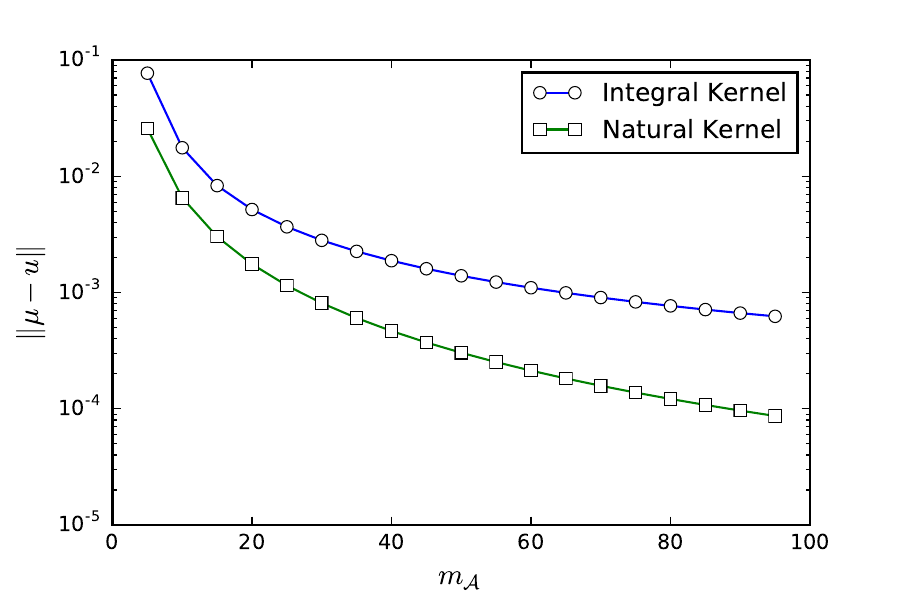}
		\caption{Error in conditional mean, $\|\mu - u\|_2$}
	\end{subfigure}
	~
	\begin{subfigure}[t]{0.45\textwidth}
		\includegraphics[width=\textwidth]{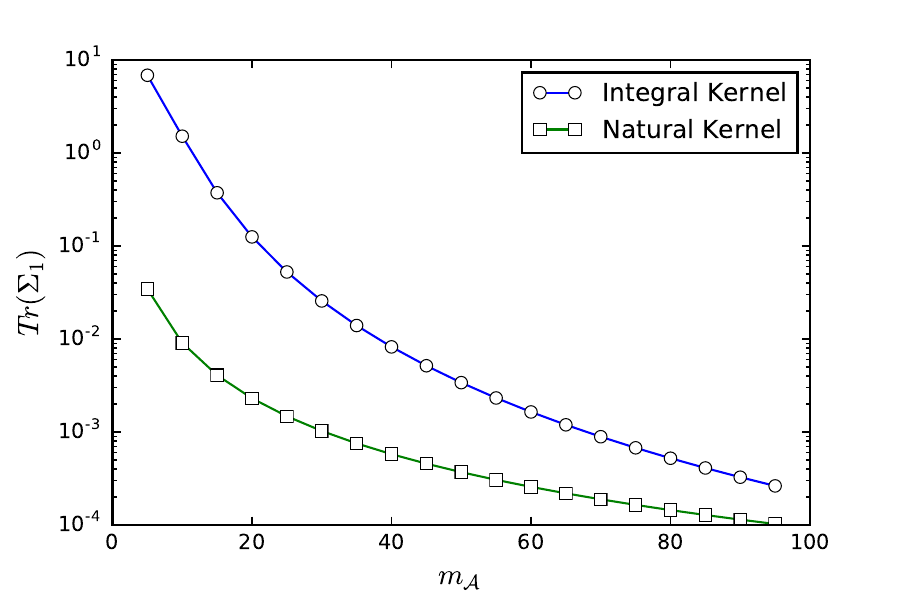}
		\caption{Residual uncertainty, $\|\sigma^2\|_1$}
	\end{subfigure}
	\caption{Probabilistic meshless methods: Convergence of mean and covariance as the number $m_{\mathcal{A}}$ of design points is increased. Values of $\mu$ and $\sigma$ were computed on a fine grid.}
	\label{fig:kernel_convergence}
\end{figure}

\subsection{The Inverse Problem}

Measurement data $\bm{y}$ are linked to the solution $u$ of the PDE through a likelihood $\pi(\bm{y} | u)$.
For linear PDEs with Gaussian additive noise $y_i = u(\bm{x}_i) + \epsilon_i$, $i = 1,\dots,n$, $\bm{\epsilon} \sim N(\bm{0} , \bm{\Gamma})$, the likelihood has the closed-form expression under marginalisation of the distributional output of the PNM:
\begin{align}
& \pi(\bm{y}|\bm{g},\bm{b},\theta) = \int \pi(\bm{y}|u) \wrt \Pi_u^{\bm{g},\bm{b}} \notag \\
& \quad = \frac{1}{\sqrt{\text{det}[2\pi(\bm{\Sigma}(\theta) + \bm{\Gamma})]}} \exp\left\{-\frac{1}{2}(\bm{y} - \bm{\mu}(\theta))^\top  (\bm{\Sigma}(\theta) + \bm{\Gamma})^{-1}(\bm{y} - \bm{\mu}(\theta))\right\}, \label{eq:marginalise}
\end{align}
where $\bm{\mu}$, $\bm{\Sigma}$ are given in Eq.~\eqref{eq:full_posterior_mean} and Eq.~\eqref{eq:full_posterior_cov}, and their dependence on the parameter $\theta$ has been emphasised.
In Section~\ref{sec:inverse_consistency} we prove that as the design $X_0$ is refined $\pi(\bm{y} | \bm{g}, \bm{b} , \theta)$ converges to the abstract likelihood $\pi(\bm{y} | \theta)$, and also prove convergence of the implied posterior distributions in an appropriate probability metric.

\subsubsection{Illustrative Example: Inverse Problem}

Returning to the example of Section~\ref{sec:illustrative_forward}, we illustrate the effect of probabilistic solution of the forward problem in terms of the inferences being made on $\theta$.
Consider the problem of estimating $\theta$ in the following system:
\begin{alignat*}{2}
	-\nabla \cdot (\theta \nabla u(x)) &= g(x) & \quad \text{for} \; & x \in (0,1) \numberthis \label{eq:illustrative_laplacian}\\
	u &= 0 &\quad \text{for} \;& x \in \{ 0, 1 \} .
\end{alignat*}
Again, take $g(x) = \sin(2\pi x)$.
Observed data for the inverse problem was generated with parameter $\theta^{\dagger} = 1$, at the locations $x=0.25$ and $x=0.75$, by evaluating the explicit solution $u(x) = (2\pi)^{-2}\sin(2\pi x)$ and corrupting these observations with additive zero-mean Gaussian noise with covariance $\bm{\Gamma} = 0.001^2 \bm{I}$.

To illustrate the advantage of a probabilistic solution to the PDE, posteriors were computed based on the standard approach of symmetric collocation --- which ignores discretisation error and replaces $\bm{u}$ with a numerical estimate --- and the PMM with data-likelihood as given in Eq.~\eqref{eq:marginalise}.
These are also contrasted with an approach in the data-likelihood covariance is inflated using a classical error estimate for symmetric collocation, as given in Proposition~\ref{prop:natural_local_accuracy}.
In this case, the data-likelihood is given by:
\begin{align*}
	\pi_{\text{diag}}(\bm{y} | \bm{g}, \bm{b}, \theta) &= \frac{1}{Z(\theta)} \exp\left\{-\frac{1}{2}(\bm{y} - \bm{\mu}(\theta))^\top  (\diag(\bm{\Sigma}(\theta)) + \bm{\Gamma})^{-1}(\bm{y} - \bm{\mu}(\theta))\right\} \\
	Z(\theta) &= \sqrt{\text{det}[2\pi(\diag(\bm{\Sigma}(\theta)) + \bm{\Gamma})]} \numberthis\label{eq:perturbed_likelihood}
\end{align*}
where $\diag(\bm{A})$ denotes the matrix whose diagonal entries are the diagonal entries of $\bm{A}$, and whose off-diagonal entries are zero.

The parameter $\theta$ was endowed with a standard log-Gaussian prior $\Pi_{\theta}$ to ensure positivity.
Figure~\ref{fig:example_posterior} shows the posteriors, while Figure~\ref{fig:example_posterior_convergence} shows convergence as the number of design points is increased. 

\begin{figure}
	\centering
	\begin{subfigure}{0.8\textwidth}
		\includegraphics[width=\textwidth]{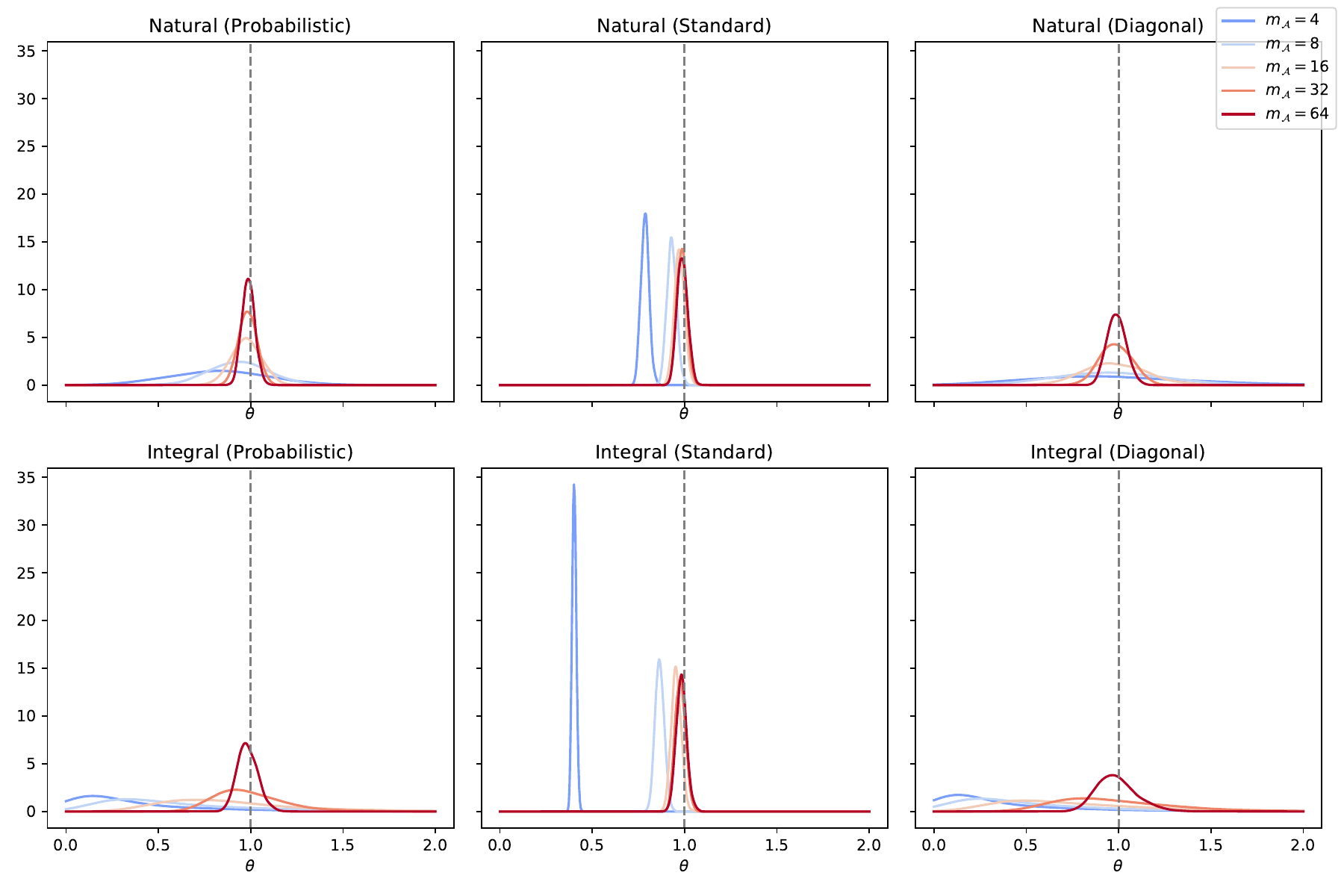}
		\caption{Posterior distributions $\Pi_\theta^{\bm{y}}$}
		\label{fig:example_posterior}
	\end{subfigure}
	\\
	\begin{subfigure}{0.8\textwidth}
		\includegraphics[width=\textwidth]{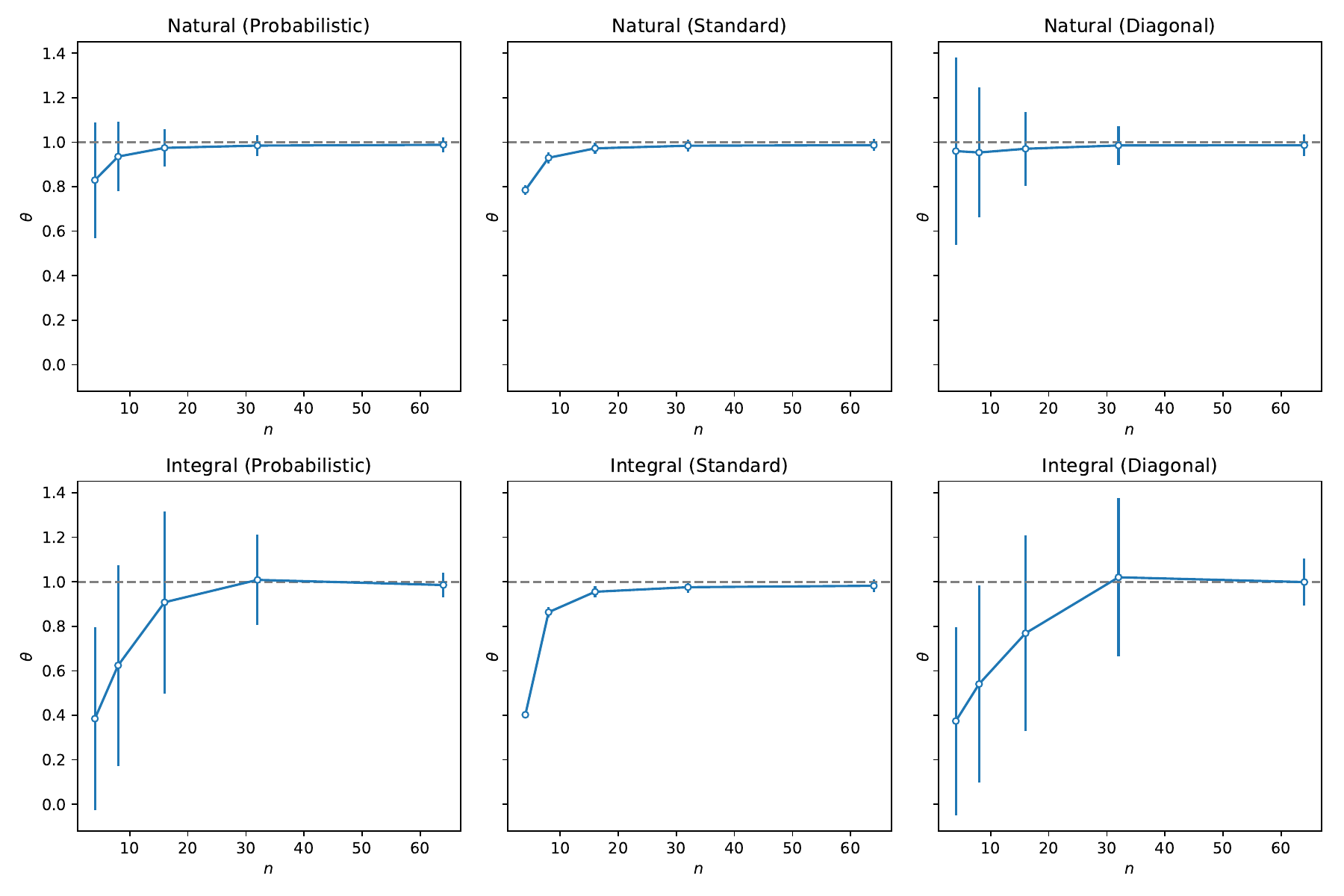}
		\caption{Posterior credible intervals}
		\label{fig:example_posterior_convergence}
	\end{subfigure}
\caption{PMM solutions to the inverse problem versus the standard approach that plugs a discrete approximation to the exact PDE solution into the likelihood, and a conservative approach which inflates the data-likelihood covariance by a classical error estimate for the forward problem.
The latter leads to over-confidence since uncertainty due to discretisation of the original PDE is ignored. 
(a) Posterior distributions $\Pi_\theta^{\bm{y}}$ as a function of the number $m_{\mathcal{A}}$ of design points.
(b) One standard deviation posterior credible intervals for $\theta$, again as a function of $m_{\mathcal{A}}$.
In each case the natural kernel $k_\textup{nat}$ is compared to the integral-type kernel $\hat{k}$, as explained in the main text.}
\end{figure}

This example highlights a shortcoming of the standard approach. The posterior variance in Figure~\ref{fig:example_posterior} (middle) is constant, independent of the number $m_{\mathcal{A}}$ of collocation basis functions, and posterior (1 s.d.) credible intervals do not cover the true value $\theta^{\dagger} = 1$ when $m_{\mathcal{A}} \leq 20$. 
In contrast, when using the PMM in Figure~\ref{fig:example_posterior} (left) there is a clear widening in the posterior for small $m_{\mathcal{A}}$, and in general the true value of $\theta$ is within a standard deviation of the posterior mode. The same is true of the posteriors in Figure~\ref{fig:example_posterior} (right), but here the inflation of the posterior is significantly wider, and the credible intervals are far more pessimistic, owing to the fact that the covariance between errors at the observation locations is ignored.

Note that, while the kernel $\hat{k}$ is independent of the value of $\theta$ in this problem, the natural kernel depends on $\theta$ through the Green's function via $k_\textup{nat}(x, x'; \theta) = \theta^{-2} k_\textup{nat}(x, x'; 1)$.
This dependence could be removed by simply dividing Eq.~\eqref{eq:illustrative_laplacian} by $\theta$, but is emphasised here to highlight theoretical considerations in Section~\ref{sec:backward_error}, in which assumption (A2) is on the strength of dependence of $k$ on $\theta$.
Note also that posterior variance is reduced when using the natural kernel $k_\textup{nat}$, compared with the integral kernel $\hat{k}$. This is to be expected considering the reduced variance exhibited in Section~\ref{sec:illustrative_forward}.

\subsubsection{Calibration of Kernel Parameters} \label{sec:calib}

The requirement to posit a kernel $\tilde{k}$ typically introduces nuisance parameters.
This issue has so far received little attention in the literature on meshless methods (e.g. \citep{Xu2017}), but is crucial to this work since the choice of parameters directly influences the spread of the probability model for numerical error.

The problem of selecting kernel parameters appears regularly in PNMs \citep{Briol2016, Conrad2015, Kersting2016}.
One approach to selection of kernel parameters would simply be to maximise the likelihood of the data, i.e.~to maximise Eq.~\eqref{eq:marginalise} over all nuisance parameters.
This is known in statistics as ``empirical Bayes'', but other approaches are possible, including cross-validation; see the discussion in \cite{Briol2016}.
Application of this approach in the context of inverse problems is difficult as the optimal kernel parameters will often be dependent upon the parameter $\theta$.
A sampling strategy for exploring posterior distributions over $\theta$ may perform poorly if kernel parameters are estimated by empirical Bayes based upon a poor initial guess for $\theta$.

An alternative approach is to consider the nuisance parameters of the kernel as additional parameters to be inferred in a hierarchical approach. This allows marginalisation of those parameters so that inferences do not depend upon point estimates, and is the approach employed in the applications in Section~\ref{sec:applications}.

\section{Theoretical Results}

The mathematical properties of the PMM method are now established.

\subsection{Error Analysis for the Forward Problem} \label{sec:forward_error}

First, we present error analysis for the forward problem where the parameter $\omega \in \Omega$ describing the differential operator is fixed. In this section we assume that the kernel in use is the kernel $\hat{k}$ of Section~\ref{sec:posited_kernel}.

Denote by $u_\omega = u(\quark, \omega, \zeta^\dagger)$ the solution of the PDE for a particular value of $\omega$, and let $u^\dagger = u_{\omega^\dagger} = u(\quark,\omega^\dagger,\zeta^\dagger)$ be the true solution to the PDE for the forcing and boundary functions $g(\quark,\zeta^\dagger)$, $b(\quark,\zeta^\dagger)$.
Here $u_\omega$ can be thought of as the solution to the PDE for a fixed value of the parameter $\theta = \theta(\omega)$, while $u^\dagger$ is the true solution to the PDE when the true value of the parameter $\theta^\dagger = \theta(\omega^\dagger)$ is used.

Two Hilbert spaces $H$, $H'$ are said to be \emph{norm-equivalent} when each is continuously embedded in the other.
We denote equivalence of $H$ and $H'$ by $H \equiv H'$.
We will work under the following assumption:
\begin{enumerate}
	\item[(A1)] Suppose that $H(D)$ is norm-equivalent to the Sobolev space $\mathbb{H}^\beta(D)$ of order $\beta > d/2$, with norm denoted by $\|\quark\|_{\mathbb{H}^\beta(D)}$.
\end{enumerate}
This can be satisfied by construction since we are free to select the kernel $\tilde{k}$.
It is implicitly assumed that the differential orders\footnote{i.e.~number of derivatives} of $\mathcal{A}$ and $\mathcal{B}$ are $O(\mathcal{A}), O(\mathcal{B}) < \beta - d/2$, so that the stochastic processes $\mathcal{A} u$ and $\mathcal{B} u$ are well-defined.

The analysis below is rooted in a dual relationship between the posterior variance and the worst-case error:
\begin{proposition}[Local accuracy] \label{prop:natural_local_accuracy}
For all $\bm{x} \in D$ we have $| \mu_{\omega}(\bm{x}) - u_\omega(\bm{x}) | \leq \sigma_\omega(\bm{x}) \norm{u_\omega}_{\hat{k}} $.
\end{proposition}
Proposition~\ref{prop:natural_local_accuracy} shows that minimising $\sigma_\omega(\bm{x})$ leads to accurate estimates $\mu_\omega(\bm{x})$.
This reassures us that the conditional measure $\Pi_u^{\bm{g},\bm{b}}$ over the solution space is locally well-behaved.
Here we have made the dependence of the functions $\mu$ and $\sigma$ on $\omega \in \Omega$ explicit.

To make precise the notion of minimising $\sigma_\omega(\bm{x})$, define the \emph{fill distance} $h = h(X_{0})$ of the design $X_0$ as
\begin{equation*}
h := \sup_{\bm{x} \in D} \min_{\bm{x}' \in X_0} \|\bm{x} - \bm{x}'\|_2.
\end{equation*}
The following is Lemma~3.4 of \cite{Cialenco2012}; see also Sections 11.3 and 16.3 of \cite{Wendland2004}:
\begin{proposition} \label{prop:sigma}
For all $\bm{x} \in D$ and all $h > 0$ sufficiently small, we have $\sigma_\omega(\bm{x}) \leq C_{\omega}^F h^{\beta - \rho - d/2}$ where $\rho = \max\{O(\mathcal{A}) , O(\mathcal{B})\}$ and $C_{\omega}^F$ is a constant dependent on $\omega \in \Omega$.
\end{proposition}
That the constant here should depend upon $\omega$ is natural, as the posterior covariance in Eq.~\eqref{eq:full_posterior_cov} depends upon both $\mathcal{A}[\omega]$ and $\mathcal{B}[\omega]$. When analysing the forward problem alone this detail is unimportant; it will, however, become relevant later.

Proposition~\ref{prop:sigma} is used to establish contraction of the conditional measure $\Pi_u^{\bm{g},\bm{b}}$ over $H(D)$ to the true solution $u_\omega$ as the fill distance is decreased:
\begin{theorem}[Contraction of conditional measure to $u_\omega$] \label{thm:fwd_contraction}
For fixed $\epsilon > 0$,
\begin{equation*}
	\Pi_u^{\bm{g},\bm{b}}\{u \in H(D) : \|u - u_\omega\|_2^2 > \epsilon\} = O\left( \frac{h^{2\beta - 2\rho - d}}{\epsilon} \right) \text{ as $h \to 0$.}
\end{equation*}
\end{theorem}
A similar result was presented as Lemma~3.5 in \cite{Cialenco2012}.
Theorem~\ref{thm:fwd_contraction} shows that the conditional measure $\Pi_u^{\bm{g},\bm{b}}$ provides sensible uncertainty quantification in a global sense.
However, the ultimate goal is to make accurate inferences on $\theta$, which introduces several considerations that go beyond analysis of the forward problem.
In particular, in the inverse problem, the solution $u = u(\quark,\omega,\pnparam)$ depends on $\omega \in \Omega$.
Thus the term $\|u_\omega\|_{\hat{k}}$ in Proposition~\ref{prop:natural_local_accuracy} will be a random variable.
From a broader perspective we must examine whether, and in what sense, the solution $u(\quark,\omega,\pnparam)$ exists as a random object.
These points are addressed in the next sections.

\subsection{Error Analysis for the Inverse Problem} \label{sec:backward_error}

The aim of Section~\ref{sec:existence_of_RV} is to establish the existence of $u(\quark,\omega,\pnparam)$ as a random object. Section~\ref{sec:inverse_consistency} shows that when using a PMM forward solver, posterior distributions for $\theta$ converge appropriately to the true posterior as $h \to 0$.

\subsubsection{Existence of the Doubly Stochastic Solution} \label{sec:existence_of_RV}

This section makes precise the sense in which the doubly stochastic solution $u(\quark,\omega,\pnparam)$ exists as a random variable.
Let $\mathbb{E}_\Omega$ and $\mathbb{E}_{\pnspace}$ denote expectations with respect to $\mathbb{P}_\Omega$ and $\mathbb{P}_{\pnspace}$.
First recall the notion of a ``Hilbert scale'' of spaces.

Consider an orthonormal basis for $H_\textup{nat}(D)$ such that a generic element $u \in H_\textup{nat}(D)$ can be written as $u = \sum_{i=1}^\infty c_i h_i$, $h_i = \sqrt{\lambda_i} e_i$ where $\lambda_1 \geq \lambda_2 \geq \dots > 0$ are eigenvalues and $e_i$ are associated eigenvectors of the integral operator $u(\quark) \mapsto \int_D u(\bm{x}) k_\textup{nat}(\bm{x},\quark) \wrt\bm{x}$.
The norm for this space is characterised by $\|u\|_\textup{nat}^2 = \sum_{i=1}^\infty c_i^2$.
Define the scale of Hilbert spaces
$H^t := \{ h = \sum_i c_i h_i \; \text{s.t.} \; \|h\|_{\textup{nat},t}^2 := \sum_i \lambda_i^{-t} c_i^2 < \infty \}$
for $t \in \mathbb{R}$ \cite[Section 7.1.3]{Dashti2014}. 
For a generic RKHS $H$ we have that $H^0 = H$, while $H^s \supseteq H^t$ whenever $s \leq t$.
The intuition here is that $H^t$, $t < 0$, is a relaxation of $H$.

An assumption is now made on the regularity of the inverse problem, as captured by the regularity of the natural solution space $H_\textup{nat}(D)$ from Section~\ref{sec:natural_space}.
Recall that $H_\textup{nat}(D)$ is a random space, depending on $\omega \in \Omega$ through the Greens function of the PDE as in  Eq.~\eqref{eq:k_definition}.
Write $\lambda_i^\textup{nat}$ for the eigenvalues associated with $k_\textup{nat}$.
Similarly write $\lambda_i^{(\alpha)}$ for the eigenvalues associated with the Sobolev space $\mathbb{H}^\alpha(D)$ of order $\alpha$.
\begin{enumerate}
\item[(A2)] For some $\alpha \geq \beta$, all $-1 < t < -d/2\alpha$ and $\mathbb{P}_\Omega$-almost all $\omega \in \Omega$, there exist constants $0 < C_{\omega}$ and $C_{\omega,t} < \infty$ such that, for all $v \in H_\textup{nat}(D)$, $i \in \mathbb{N}$,
\begin{equation*}
	\|v\|_{\mathbb{H}^\alpha(D),t}^2 \; \leq \; C_{\omega,t} \|v\|_{\textup{nat},t}^2, \; \; \; \lambda_i^{(\alpha)} \; \leq \; C_\omega \lambda_i^\textup{nat}
\end{equation*}
and $\mathbb{E}_\Omega[\; C_{\omega}^{t} \; C_{\omega,t} \;] \; < \; \infty$.
\end{enumerate}
(A2) implies, in particular, that $H_\textup{nat}^t(D)$ is embedded in $[\mathbb{H}^\alpha(D)]^t \equiv \mathbb{H}^{(1+t) \alpha}(D)$ for $\Pi_\theta$-almost all values of the parameter $\theta \in \Theta$.
Note that $\alpha$ and $\beta$ assume distinct roles in the analysis; $\alpha$ captures the regularity of the unavailable natural solution space $H_\textup{nat}(D)$, while $\beta$ captures the regularity of the larger space $H(D)$ in which the numerical solver operates.
In general we must have $\alpha \geq \beta$.

\begin{theorem}[Existence] \label{thm:well_defined}
For all $0 < s < \alpha - d/2$, the function $u$ exists as a random variable in 
\begin{equation*}
L_{\mathbb{P}_\Omega \mathbb{P}_{\pnspace}}^2(\Omega , \pnspace; \mathbb{H}^s(D)) := \{ v \colon D \times \Omega \times \pnspace \to \mathbb{R} \; \text{\normalfont s.t.} \; \mathbb{E}_\Omega \mathbb{E}_{\pnspace} \|v\|_{\mathbb{H}^s(D)}^2 < \infty \}
\end{equation*}
and takes values $(\mathbb{P}_\Omega,\mathbb{P}_{\pnspace})$-almost surely in $\mathbb{H}^s(D)$.
\end{theorem}

The sense of existence used in Theorem \ref{thm:well_defined} is precisely the same sense in which a Gaussian process exists, where the covariance function forms a kernel for $\mathbb{H}^\alpha(D)$ \cite[Theorem 2.10]{Dashti2014}.

\subsubsection{Posterior Contraction} \label{sec:inverse_consistency}

This section elaborates on the sense in which the posterior in the Bayesian inverse problem approaches the idealised posterior $\Pi_\theta^{\bm{y}}$ in Eq.~\eqref{eq:bayes_first}, when the forward problem is approximated using the PMM.

Suppose that the data $\bm{y}$ are noisy measurements of the solution $u(\bm{x})$ at locations $X = \{\bm{x}_j\}_{j=1}^n \subset D$. Further assume that measurement noise is Gaussian, so that $\bm{y} \sim N(\bm{u} , \bm{\Gamma})$.
Define the potential $\Phi_h(\bm{y} , \theta) := -\log \pi(\bm{y} | \theta , \bm{g}, \bm{b})$, to emphasise dependence on the fill distance $h$. 
Note that many sets of collocation points $X_0$ can each have the same fill distance, so $h$ does not uniquely define $\Phi_h$.
Recall that the posterior distribution $\Pi_\theta^{\bm{y}, h}$ is given by
\begin{align*}
\frac{\wrt \Pi_\theta^{\bm{y},h}}{\wrt\Pi_\theta}(\theta) = \frac{1}{Z_h} \exp(- \Phi_h(\bm{y} , \theta)) , \numberthis \label{norm BIP 2}
\intertext{where}
Z_h := \int_\Theta \exp(-\Phi_h(\bm{y} , \theta)) \Pi_\theta(\wrt \theta) .
\end{align*}
Following \cite[Section 3.4.3]{Dashti2014}, it is clear that $Z_h > 0$ for all $h$, provided that $\mathcal{A}$ and $\mathcal{B}$ are non-degenerate, so that the conditional covariance $\Sigma(\theta)$ is finite and positive-definite for all $\theta \in \Theta$, and that the set $X_0$ contains only unique points. Thus, the posterior distribution is well-defined for all $h$.

Now of interest is whether the posterior distribution $\Pi_\theta^{\bm{y}, h}$ contracts to $\Pi_\theta^{\bm{y}}$ as $h\to 0$, as quantified by the Hellinger metric $d_\textup{Hell}$ given by
\begin{equation*}
	d_\textup{Hell}(\Pi_{\theta}^{\bm{y}, h}, \Pi_\theta^{\bm{y}})^{2} := 1 - \int_\Theta \left( \frac{\wrt \Pi_\theta^{\bm{y}, h}}{\wrt \Pi_\theta^{\bm{y}}}(\theta) \right)^{1/2} \, \Pi_\theta^{\bm{y}}(\wrt \theta) .
\end{equation*}
This requires an additional assumption:
\begin{enumerate}
	\item[(A3)] There exists a function $C(\norm{\theta}_\Theta)$ such that for each $\omega$, 
	\begin{equation*}
		\max \{ C_\omega^F , C_\omega^F \norm{u_\omega}_{\hat{k}}, C_\omega^F \norm{u_\omega}_{\hat{k}}^2 \} \leq C(\norm{\theta(\omega)}_\Theta)
	\end{equation*} 
	and $\int C(\|\theta\|_{\Theta})^4 \wrt \Pi_\theta < \infty$.
\end{enumerate}
Then, the following theorem follows:
\begin{theorem}[Robustness to Approximation Error] \label{thm:theta_consistency}
The posterior distribution $\Pi_\theta^{\bm{y}, h}$ satisfies $d_\textup{Hell}(\Pi_\theta^{\bm{y}, h}, \Pi_\theta^{\bm{y}}) = \mathcal{O}(h^{\beta - \rho - d/2})$.
\end{theorem}

\section{Implementation} \label{sec:compute}

In this section several computational details are considered. First, in Section~\ref{sec:experimental_design} the selection of the set $X_0$ of design locations is discussed. Second, a method for extending the procedure described above to a limited class of nonlinear systems is proposed in Section~\ref{sec:nonlinear}.

\subsection{Selection of Collocation Points} \label{sec:experimental_design}

The inferences drawn with PMMs are valid from a statistical perspective, regardless of the locations $X_0$ that are used to implement the method.
Yet, how informative the inferences are will depend upon choice of $X_0$.
This section outlines a principled approach to selecting $X_0$ based on experimental design. 

The selection of $X_{0}$ is subject to competing considerations.
The theoretical results of Section~\ref{sec:forward_error} imply that a principal goal in minimising error is to minimise the fill distance $h(X_{0})$.
However, the fact that data $\bm{y} \in \reals^n$ are obtained at specific locations suggests that design points should be placed to minimise uncertainty at those locations, for example by placing collocation points near to data locations.
Two additional requirements arise from a practical perspective:
First, while a large number $m = m_{\mathcal{A}} + m_{\mathcal{B}}$ of design points will minimise uncertainty, the computational cost of the method grows rapidly with $m$.
Second, the method used to \emph{select} locations must not be too computational.

The approach pursued is to cast the choice of $X_0$ as a problem of statistical experimental design.
This is made possible by the probabilistic formulation of the meshless method.
Below we write $\bm{\Sigma}(\theta,X_0) \in \reals^{n \times n}$ for the posterior covariance matrix of the $\reals^{n}$-valued solution vector $\bm{u}$ to emphasise that this depends on both the parameter $\theta$ and the choice of design $X_0$.

To proceed, let $L \colon \reals^{n \times n} \to \reals$ denote a loss function.
Define an optimal design $X_0^*(\theta)$ for fixed $\theta$ to satisfy:
\begin{equation}
  X_0^*(\theta) \in \argmin_{\substack{
    X_0^{\mathcal{A}} \subset D , \; X_0^{\mathcal{B}} \subset \partial D \\
    |X_0^{\mathcal{A}}| = m_{\mathcal{A}}, \; |X_0^{\mathcal{B}}| = m_{\mathcal{B}}}
  } L[\bm{\Sigma}(\theta,X_0)]. \label{eq:optimal_design}
\end{equation}
Particular choices of $L$ are suggested from Bayesian decision theory.
For example, an A-optimal design minimises $L[\bm{\Sigma}] = \trace[\bm{\Sigma}]$, while a D-optimal design minimises $L(\bm{\Sigma}) = \det[\bm{\Sigma}]$.
Both A- and D-optimality can be motivated from a Hilbert-space perspective on
posterior uncertainty.
The A-optimality criterion is equivalent to minimising the trace of the posterior covariance operator $\int_D \sigma(\bm{x})^2 \wrt\bm{x}$, while D-optimality minimises the volume of the uncertainty ellipsoid \citep{Alexanderian2014a}.

The minimisation required to determine optimal designs required by Eq.~\eqref{eq:optimal_design} is non-trivial, as it is both high-dimensional and non-convex. The computational details are described in the supplement, Section~\ref{sec:experiment_design_supplement}.

\subsection{Extension to a Class of Semi-linear PDEs} \label{sec:nonlinear}

An important motivation for probabilistic numerical solvers comes from inverse problems that involve a non-linear forward model. At present little is known about the performance of meshless methods in this setting.

Non-linear problems abound in the applied sciences and numerical methods for these problems require substantially more computational effort.
There is thus a strong computational motivation for exploiting meshless methods in many non-linear inverse problems.
However, the theoretical analysis of meshless methods for non-linear problems is not available and for inferences to be statistically valid, a more detailed characterisation of numerical error is required.
In this section the framework of PMM is extended to the case when case when the underlying PDE model is non-linear.
To limit scope, the focus here is on a particular class of non-linear PDEs, known as \emph{semi-linear} PDEs, that are rich enough to exhibit canonical non-linear behaviour (e.g.~multiple solutions), whilst also permitting tractable algorithms.

\subsubsection{A Latent Variables Approach for Semi-Linear PDEs}

Here we generalise the previous sections to operators of the form $\mathcal{A} = \mathcal{A}_1 + \dots + \mathcal{A}_N$, where each of the $\mathcal{A}_j$ is either a linear differential operator, or a possibly non-linear monotonic operator.
This class is motivated by the observation that monotonic operators are invertible.
Below, this invertibility is exploited to reduce the system to a linear system, to which above methods can be applied:

As an illustrative example, consider the steady-state Allen--Cahn equation, which is often used to model the boundaries between phases in alloys \cite{Allen1979}:
\begin{equation*}
	-\theta \nabla^2 u(\bm{x}) + \theta^{-1} (u(\bm{x})^3 - u(\bm{x})) = g(\bm{x}).
\end{equation*}
This is a semi-linear PDE with linear differential operator $\mathcal{A}_1 u = - \theta \nabla^2 u - \theta^{-1} u$ and monotonic operator $\mathcal{A}_2 u = \theta^{-1} u^3$, where $\theta > 0$ is a scalar parameter.

In the case of $N = 2$ we have $\mathcal{A} = \mathcal{A}_1 + \mathcal{A}_2$ and the indirect, non-linear observations $\mathcal{A} u(\bm{x}_i) = g(\bm{x}_i)$ can be decomposed into direct observations by introducing a latent function $z$ such that $\mathcal{A}_1 u(\bm{x}_i) = z(\bm{x}_i)$ and $\mathcal{A}_2 u(\bm{x}_i) = g(\bm{x}_i) - z(\bm{x}_i)$. 
Concretely, for the Allen--Cahn system,
\begin{align*}
	-\theta \nabla^2 u(\bm{x}_i) - \theta^{-1} u(\bm{x}_i) &= z(\bm{x}_i) \\
	\theta^{-1} u(\bm{x}_i)^3 &= g(\bm{x}_i) - z(\bm{x}_i).
\end{align*}
The final equation can be inverted to produce $u(\bm{x}) = (\theta(g(\bm{x}_i) - z(\bm{x}_i)))^{1/3}$, which leads to a system of equations that is linear in the solution $u$, depending on the unknown function $z$, and can be solved using the methods previously introduced.
However, to make inferences on both the actual solution $u$ in the forward problem, and $\theta$ in the inverse problem, we must be able to efficiently marginalise the unknown latent function $z$.
This is discussed next.

\subsubsection{Conditional Measure with Latent Variables} \label{sec:semi_linear_conditional}

To simplify the notation, details are presented here for only the simplest case, in which $\mathcal{A}_1$ is a linear differential operator and $\mathcal{A}_2$ is a monotonic operator such that $\mathcal{A}_2^{-1}$ is known.
The previous notation is extended as follows:
\begin{equation*}
\bm{z} = \mathcal{A}_1 \bm{u}, \; \; \; \mathcal{L} = \begin{bmatrix} \mathcal{A}_1 \\ \mathcal{I} \\ \mathcal{B} \end{bmatrix}, \; \; \; \bar{\mathcal{L}} = \begin{bmatrix} \bar{\mathcal{A}}_1 & \mathcal{I} & \bar{\mathcal{B}} \end{bmatrix} ,
\end{equation*}
where $\mathcal{I} \colon H(D) \to H(D)$ is the identity operator.
Here we have written $\bm{z}$ for the $m_{\mathcal{A}} \times 1$ vector with $j$\textsuperscript{th} element $z(\bm{x}_i)$.
To simplify the notation in this section, dependence on the parameter $\theta$ is suppressed.
For non-linear $\mathcal{A}_2$, the marginal probability distribution $\Pi_u^{\bm{g},\bm{b}}$ will no longer be Gaussian.
However, when $\bm{z}$ is included, $\Pi_u^{\bm{z},\bm{g},\bm{b}}$, representing the conditional distribution of $u | \bm{z}, \bm{g}, \bm{b}$, is Gaussian and its finite dimensional distribution at the test points $X$ takes the form
\begin{equation*}
\bm{u} | \bm{z}, \bm{g} , \bm{b} \sim N(\bm{\mu} , \bm{\Sigma})
\end{equation*}
where
\begin{align*}
\bm{\mu} & =  \bar{\mathcal{L}}\bm{K}(X,X_0) [\mathcal{L}\bar{\mathcal{L}}\bm{K}(X_0)]^{-1} \begin{bmatrix}\bm{z} \\ \mathcal{A}_2^{-1}(\bm{g} - \bm{z}) \\ \bm{b} \end{bmatrix} \\
\bm{\Sigma} & = \bm{K}(X) - \bar{\mathcal{L}}\bm{K}(X,X_0) [ \mathcal{L}\bar{\mathcal{L}}\bm{K}(X_0)]^{-1} \mathcal{L}\bm{K}(X_0,X) .
\end{align*}
This observation suggests that an efficient sampling scheme for $\Pi_u^{\bm{z} , \bm{g}, \bm{b}}$ can be constructed, with samples from $\Pi_u^{\bm{g}, \bm{b}}$ read off as a marginal.
Full details are provided in the supplement, Section~\ref{sec:pm_mcmc}.
In Section \ref{sec:allen_cahn} we present numerical experiments that make use of this latent variables approach.

\section{Experiments} \label{sec:applications}

The empirical performance of PMM is now explored.
The first application considers electrical impedance tomography (EIT), an infinite-dimensional inverse problem with linear governing equations.
The second is a more challenging application to the steady-state Allen--Cahn equation, a finite-dimensional inverse problem with non-linear governing equations.
These two applications combine to illustrate the salient properties of the method.\footnote{The numerical experiments in this section can be reproduced using the Python library hosted at \url{https://github.com/jcockayne/bayesian_pdes}.}

\subsection{Application to Electrical Impedance Tomography} \label{sec:EIT}


EIT is a technique used for medical imaging in which an electrical current is passed through electrodes attached to a patient.
The statistical challenge is to use voltage measurements from these electrodes to determine interior conductivity, for example for the purposes of detecting brain tumours. 
It is known that EIT is well-posed as a PDE-constrained Bayesian inverse problem \cite{Dunlop2015}.
However, previous work required numerical error in the PDE solver to be tightly controlled, at increased numerical cost (e.g. \citep{Schwab2012}).

The nature of the observations, coupled with the complexity of the conductivity field, makes sampling from the posterior in the Bayesian inverse problem difficult. 
PMMs are attractive as they permit a coarse discretisation to be used while still providing rigorous statistical inference. 
Below it is shown that when the conductivity field is recovered using a cheap PMM for the forward problem, the posterior variance of the field is appropriately inflated and remains meaningful.

In this work a simplified version of the EIT model is used, as originally posed in \cite{Calderon1980}, in which it is assumed that a current is applied on $\partial D$, while the electrodes are modelled as points. 
The system takes the form of an elliptic PDE that is linear in the solution $u$, the voltage relative to an arbitrary ground voltage:
\begin{alignat}{2}
	-\nabla \cdot (a(\bm{x}) \nabla u(\bm{x})) &= 0 &\quad& \bm{x} \in D \label{eq:eit_interior}\\
	a(\bm{s}_e) \frac{\partial u}{\partial \bm{n}}(\bm{s}_e) &= c_{e} &\quad& e=1,\dots,N_S \label{eq:eit_boundary}
\end{alignat}
where $a$ is the unknown conductivity field, to be recovered, $\frac{\partial}{\partial \bm{n}}$ denotes the directional derivative along the outward pointing normal vector, $\set{\bm{s}_e} \subset \partial S$, $e=1,\dots,N_s$ denote the locations of the point electrodes on the boundary, and $c_{e}$ denotes the current applied to electrode $e$. 

\begin{figure}
	\centering
	\includegraphics[width=0.5\textwidth]{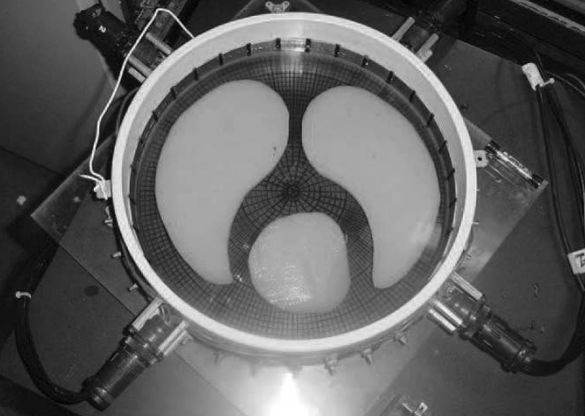}
	\caption{Agar targets from which the measurements used in Section~\ref{sec:EIT} were obtained. The two large lung-shaped targets each have a lower conductivity than the surrounding saline, while the smaller heart-shaped target has a higher conductivity.}
	\label{fig:eit_experiment_practical}
\end{figure}

The inverse problem is to infer $a$ from measurements $y_e^{(j)}$ of the sensor voltages $u(\bm{s}_e)$ obtained from various current patterns $\bm{c}^{(j)}$. The sensor voltages are assumed to have been corrupted with Gaussian noise with variance equal to $5.0$, chosen arbitrarily based on the scale of the measurements owing to a lack of information on the actual measurement error of the sensors.
Data was obtained from the EIDORS suite of contributed data\footnote{This can be found online at \url{http://eidors3d.sourceforge.net/data_contrib/jn_chest_phantom/jn_chest_phantom.shtml}.} and is due to \cite{Isaacson2004}. 
This data was obtained from measurements of 32 equispaced electrodes around the perimeter of a circular tank filled with saline solution, into which three agar targets were placed, as depicted in Figure~\ref{fig:eit_experiment_practical}. The larger two targets each have a lower conductivity than the surrounding saline, while the smaller target has a higher conductivity.
Data consisted of direct voltage measurements from each electrode for $N_s=31$ distinct stimulation patterns. Each pattern $j$ applies a current $c_{e}^{(j)}$ to each electrode $e$ given by
\begin{equation*}
	c_{e}^{(j)} = \left\{\begin{array}{ll}
		M\cos\left( j \phi_e\right) & j < (N_s+1)/2 \\
		M\cos\left( \pi e \right) & j = (N_s+1)/2\\
		M\sin\left( \left(j - (N_s+1)/2 \right) \phi_e \right) & j > (N_s+1)/2
	\end{array} \right.,
\end{equation*}
where $M$ is the amplitude of the current, and $\phi_e = \frac{2\pi e}{(N_s+1)/2}$.
The units for stimulation were milliamps, while voltages were measured in volts.

In the implementation of the forward solver the domain $D$ was modelled as a unit disc. 
The solution was endowed with a Gaussian prior, $u \sim N(0, k)$, where $k$ was taken to be a squared exponential covariance $k(\bm{x}, \bm{x}') = \sigma \exp( - \frac{1}{2\ell^2} \norm{\bm{x} - \bm{x}'}^2)$. Here $\sigma$ is a hyper-parameter which captures the amplitude of samples from the prior and $\ell$ is a length-scale hyper-parameter. 
Note that, since the convolution of the squared exponential kernel with itself is again squared exponential, this corresponds to choosing $\tilde{k}$ to be squared exponential in Section~\ref{sec:posited_kernel}.
The parameter $\sigma$ was fixed to 100, chosen to ensure that the range of solutions corresponds to the range of boundary voltage observations, while $\ell$ was endowed with a half-range Cauchy prior, following the advice of \cite{Gelman2006}, and marginalised.
In this instance the experimental design methodology was not used, as it was found that a large number of design points were required, so that the optimisation problem for determining the optimal design was prohibitively computational. 
As a result, regular designs were used for a variety of different sizes $m_\mathcal{A}$.
Figure~\ref{fig:eit_solution_design} shows an example of such a design.

A log-Gaussian prior was assigned to the conductivity field, with $a(\bm{x}) = \log(\theta(\bm{x}))$ and $\Pi_\theta = N(m_\theta, k_\theta)$. The kernel $k_\theta$ was taken to be squared-exponential, with $\ell = 0.3$ and $\sigma = 1.0$. The prior mean was fixed to a constant, chosen by maximising the log-likelihood of observations over constant conductivity fields. 
Note that the conductivity field would not be expected to be so smooth, particularly for this problem where the agar has hard boundaries. 
Hard boundaries can be recovered using techniques such as level set inversion \cite{Dunlop:2016gb}, but a relaxation to smooth conductivity fields is common in EIT.
Furthermore this relaxation is required here because a strong-form solution for the forward problem is sought, which precludes discontinuities in the conductivity field.

The inverse problem was solved using the preconditioned Crank--Nicolson (pCN) method as described in Section~\ref{sec:pCN}.
For this purpose the conductivity field was discretised to a regular grid of 177 points, as depicted in Figure~\ref{fig:eit_conductivity_design}. 

\begin{figure}
\centering
	\begin{subfigure}{0.6\textwidth}
	\centering
		\includegraphics[height=120px]{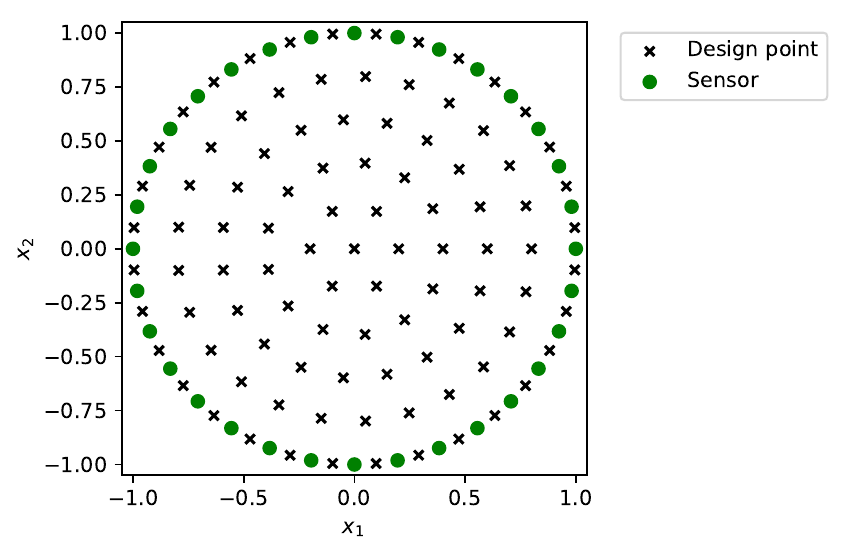}
		\caption{Discretisation of $u(\bm{x})$} \label{fig:eit_solution_design}
	\end{subfigure}
	\begin{subfigure}{0.39\textwidth}
	\centering
		\includegraphics[height=120px]{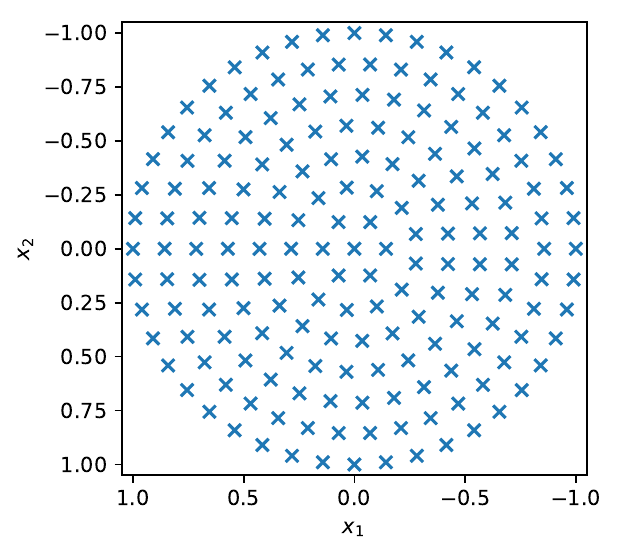}
		\caption{Discretisation of $\theta(\bm{x})$} \label{fig:eit_conductivity_design}
	\end{subfigure}
	\caption{Regular designs for the EIT experiment in Section~\ref{sec:EIT}. The design in \ref{fig:eit_solution_design} is a representative example for $m_\mathcal{A}=96$ points.}
\end{figure}

In Figure~\ref{fig:eit_posterior_means} posterior conductivity fields are plotted for $m_\mathcal{A}=96$, $127$, $165$ and $209$ design points. Each is based on $10,000,000$ iterations of pCN with the first $5,000,000$ discarded. Compared with Figure~\ref{fig:eit_experiment_practical} the qualitative accuracy of the recovered field is clear, though even with the coarsest discretisation the posterior over $\theta$ shows the main features of the targets. In Figure~\ref{fig:eit_mean_convergence} the distance between the posterior mean of the conductivity field and that of a reference conductivity field $\theta_\text{ref}$, obtained from a finer discretisation with $m_\mathcal{A}=259$, is displayed.

\begin{figure}
	\centering
	\begin{subfigure}{0.24\textwidth}
		\centering
		\includegraphics[width=\textwidth]{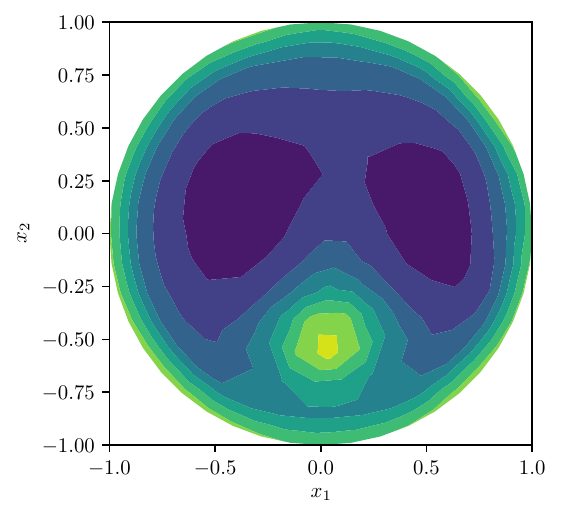}
		\caption{$m_\mathcal{A}=96$}
	\end{subfigure}
	\begin{subfigure}{0.24\textwidth}
		\centering
		\includegraphics[width=\textwidth]{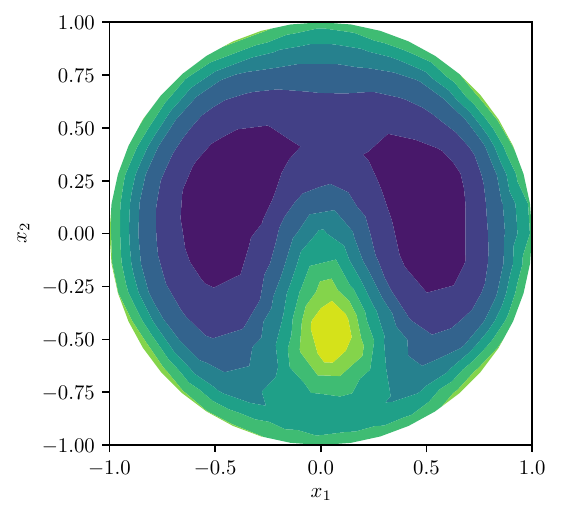}
		\caption{$m_\mathcal{A}=127$}
	\end{subfigure}
	\begin{subfigure}{0.24\textwidth}
		\centering
		\includegraphics[width=\textwidth]{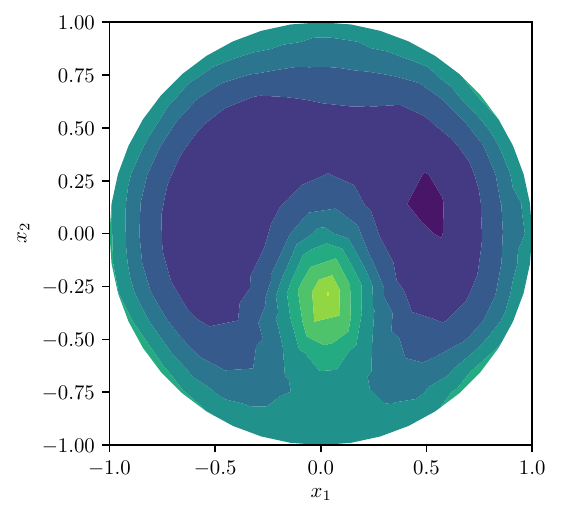}
		\caption{$m_\mathcal{A}=165$}
	\end{subfigure}
	\begin{subfigure}{0.24\textwidth}
		\centering
		\includegraphics[width=\textwidth]{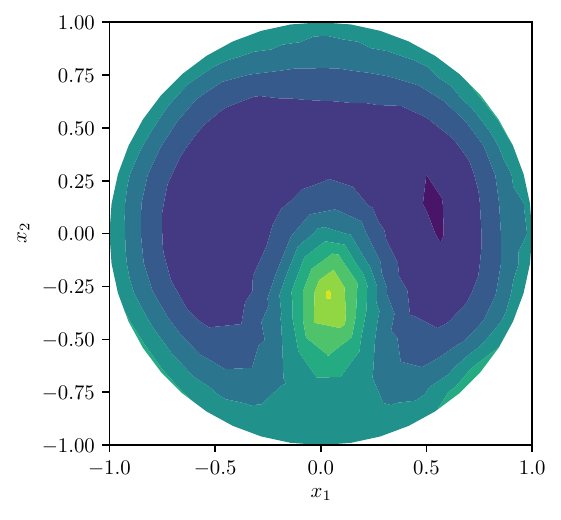}
		\caption{$m_\mathcal{A}=209$}
	\end{subfigure}
	\caption{Mean of $a(\bm{x})$ for Section~\ref{sec:EIT}. Each figure shows the posterior mean for the PMM forward solver.} \label{fig:eit_posterior_means}

	\includegraphics[width=0.5\textwidth]{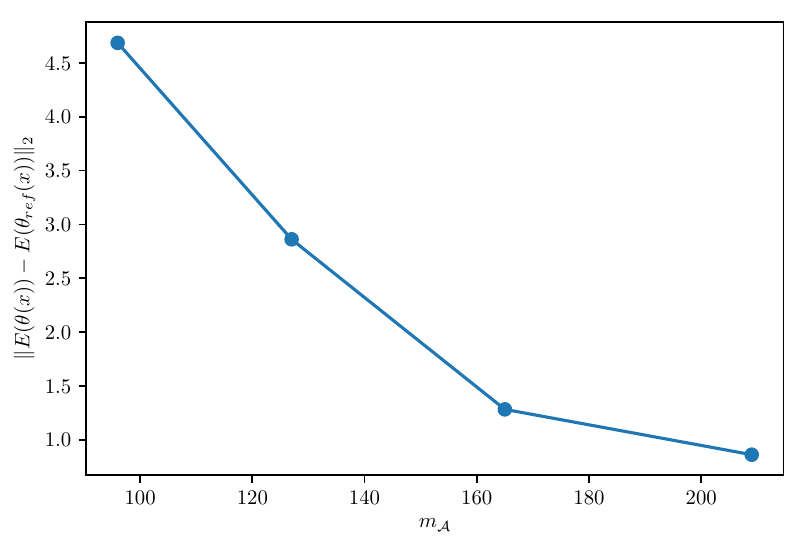}
	\caption{Convergence of the posterior mean to the reference field $\theta_\text{ref}$ in Section~\ref{sec:EIT}.}
	\label{fig:eit_mean_convergence}
\end{figure}

In Figure~\ref{fig:eit_variance_ratio} the ratio of pointwise variance of the PMM and symmetric collocation posteriors is displayed. It can be seen that almost universally throughout the domain, the PMM forward solver results in a higher variance, as expected. 
To summarise this more quantitatively we computed
\begin{equation*}
	s_{m_\mathcal{A}} := \int_{\Theta} \norm{\Sigma_\textrm{coll}^{-1/2}(\theta - \mu_{\textrm{coll}})}_2^2 \wrt \Pi_\theta^{\bm{y},h}(\theta)
\end{equation*}
where $\Pi_\theta^{\bm{y},h}$ is the posterior for $\theta$ based on PMM with $m_\mathcal{A}$ design points, while the mean $\mu_\text{coll}$ and covariance $\Sigma_\text{coll}$ are the posterior mean and variance for $\theta$ obtained instead with a symmetric collocation forward solver rather than a PMM forward solver. 
To first order, when $s_{m_\mathcal{A}} < 1$ the PMM posterior $\Pi_\theta^{\bm{y},h}$ is interpreted as being over-confident relative to the collocation benchmark; this would represent a failure to account for discretisation error. Note though that $s_{m_{\mathcal{A}}} > 1$ does not, in itself, imply that the uncertainty quantification is correct.
Figure~\ref{fig:eit_statistic_convergence} shows, as expected, that the posterior distribution is more conservative when using the PMM forward solver.
This demonstrates that more cautious inferences are being arrived at when uncertainty due to discretisation of the PDE is taken into account.
Furthermore the inferences become broadly less conservative as $m_\mathcal{A}$ is increased.

\begin{figure}
	\centering
	\begin{subfigure}{0.44\textwidth}
	\centering
	\includegraphics[width=\textwidth]{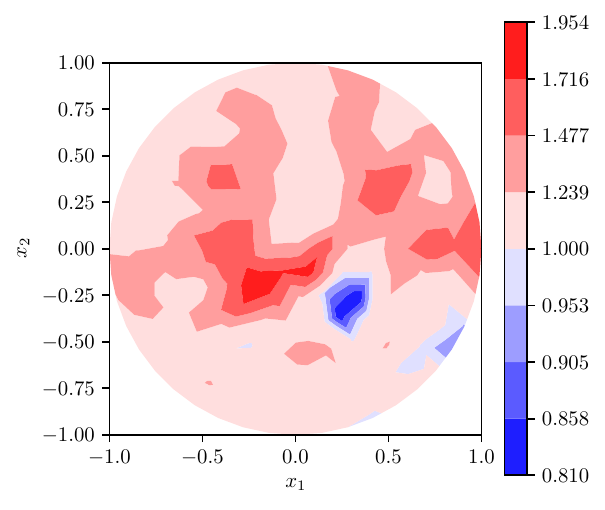}
	\caption{Ratio of the variance in the posterior distribution arising from using a PMM forward solver, compared to a symmetric collocation forward solver, at $m_\mathcal{A}=96$ points.}
	\label{fig:eit_variance_ratio}
	\end{subfigure}
	\begin{subfigure}{0.47\textwidth}
	\centering
	\includegraphics[width=\textwidth]{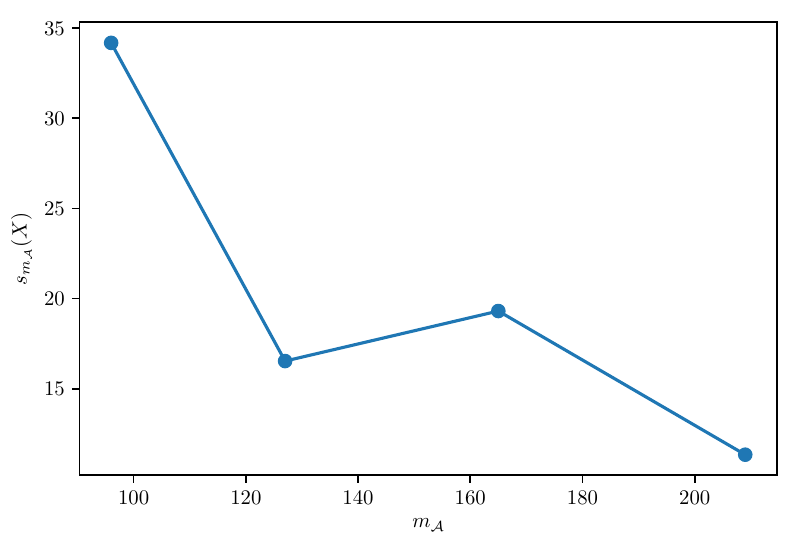}
	\caption{Convergence of the statistic $s_{m_\mathcal{A}}$ as a function of $m_\mathcal{A}$. }
	\end{subfigure}
	\caption{Posterior variance analysis for the analysis in Section~\ref{sec:EIT}.}
	\label{fig:eit_statistic_convergence}
\end{figure}

\subsection{Application to the steady-state Allen--Cahn Equation} \label{sec:allen_cahn}

Next we considered the steady-state Allen--Cahn system \citep{Allen1979}
\begin{align*}
-\theta \; \nabla^2 u + \theta^{-1} (u^3 - u) &= 0 &&x \in (0,1)^2 \\
u &= +1  &&x_1 \in \{0,1\}, \; 0 < x_2 < 1 \\
u &= -1 &&x_2 \in \{0,1\}, \; 0 < x_1 < 1 .
\end{align*}
The data-generating value $\theta^{\dagger} = 0.04$ was used, which is of interest because it leads to three distinct solutions of the PDE.
This set-up was recently considered in \cite{Farrell2014}, where the deflation technique was used to determine all solutions $u_1$ (``negative stable''), $u_2$ (``unstable'') and $u_3$ (``positive stable''), which are shown in Figure~\ref{fig:AC_sols}.
The existence of multiple solutions provides additional motivation for the quantitative description of solver error that is provided by PMMs.
Data were generated from the unstable solution $u_2$ to this system; in total $n = 16$ observations were taken on a $4\times 4$ grid in the interior of the domain and each observation was corrupted with Gaussian noise with covariance $\bm{\Gamma} = 0.1^2\bm{I}$.
The existence of three distinct solutions to the PDE for each value of $\theta$ was ensured by using a uniform prior for $\theta$ over $(0.02,0.15)$ for the inverse problem.

\begin{figure}
	\includegraphics[width=\textwidth]{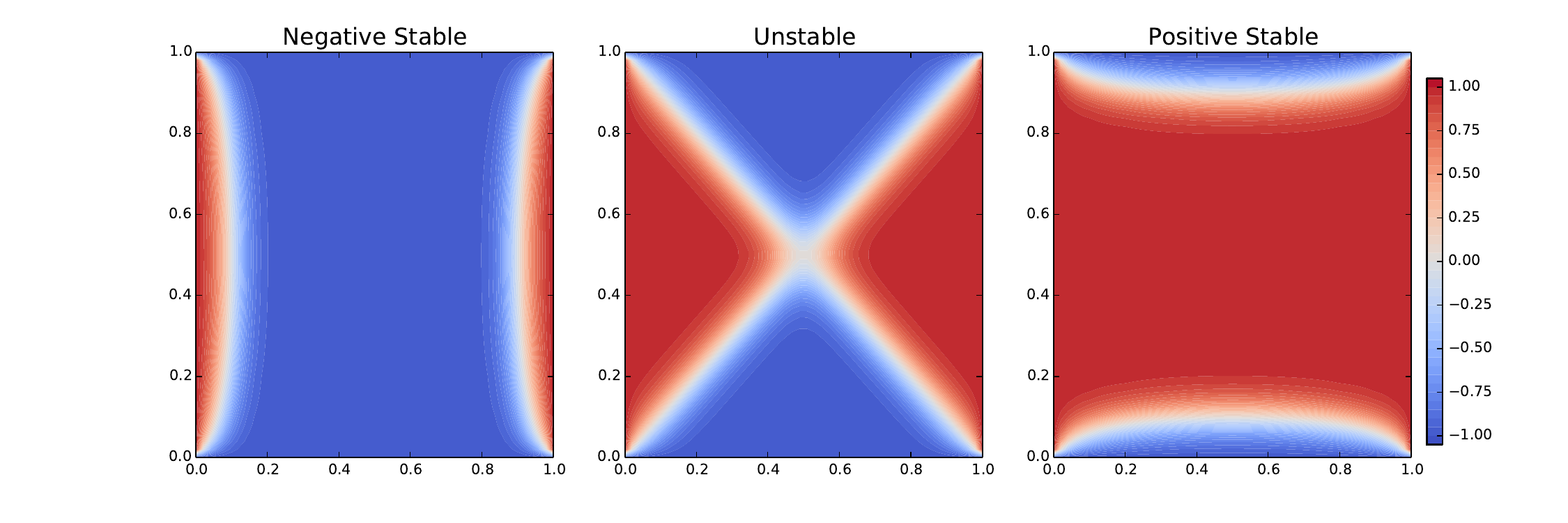}
	\caption{Three distinct solutions of the Allen--Cahn equations at $\delta^\dagger = 0.04$.}
	\label{fig:AC_sols}
\end{figure}

For the PMM, the reference measure was based on the squared-exponential kernel.
The length scale $\ell$ was assigned a standard half-range Cauchy prior and was marginalised, again following the recommendation of \cite{Gelman2006}. The variance parameter $\sigma$ was set to $\sigma=1$ to ensure that the support of the prior covers the anticipated range of the solution.

Experimental designs were computed as discussed in Section~\ref{sec:experimental_design}.
One design is shown in Figure \ref{fig:AC_designs}, in this case based on the the solution at $\theta^\dagger = 0.04$. 
The space-filling form of this design perhaps highlights an inefficiency in the assumption of isotropic covariance; in the case of the Allen--Cahn system, it is clear that the three solutions are flat in most of the domain but sharply varying in specific regions.

\begin{figure}
	\centering
	\includegraphics[width=0.8\textwidth]{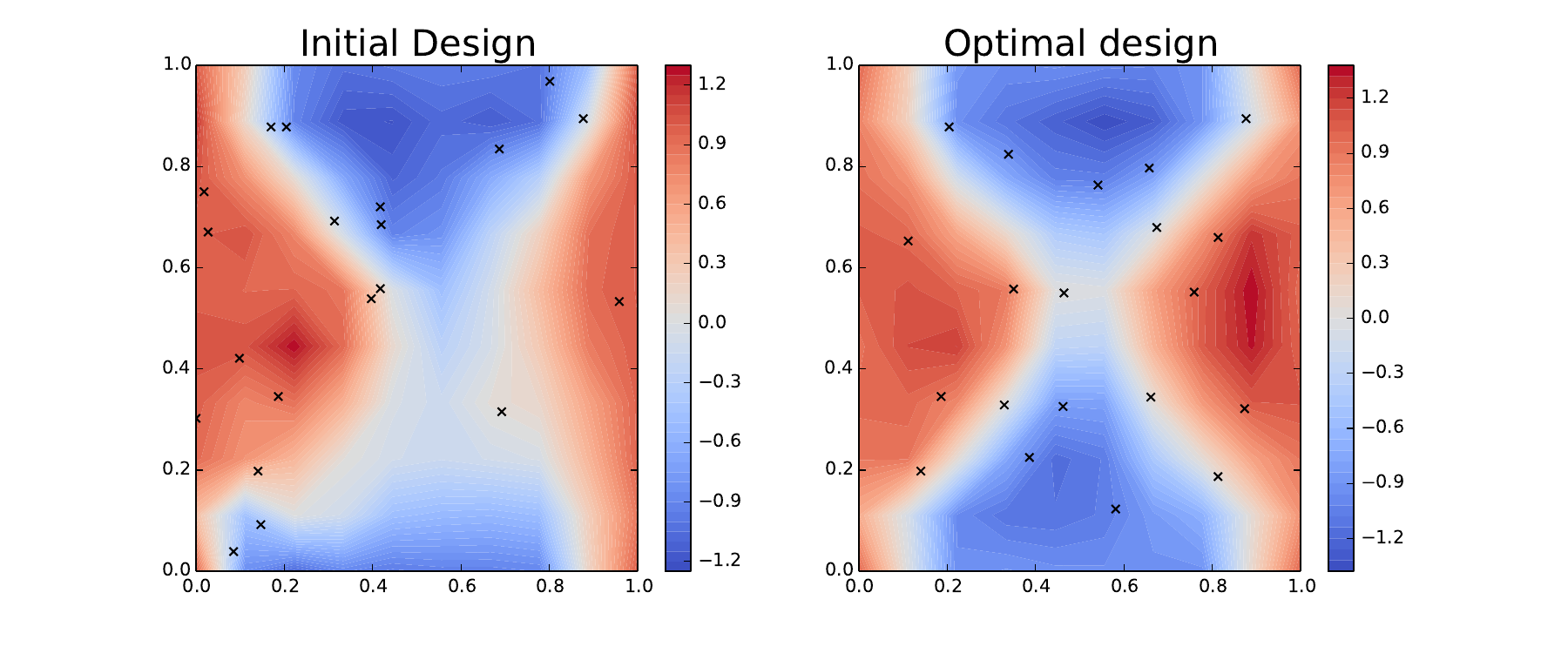}
	\caption{Application to Allen--Cahn: Initial vs. optimal design, $m_{\mathcal{A}} = 20$ points. The heat map shown is the mean function of the conditional measure $\Pi_u^{\bm{g},\bm{b}}$ for the unstable solution, the accuracy of which is controlled by the quality of the design $X_0$.
	Left: Initialisation used in the optimisation method.
	Right: Final design after optimisation.}
	\label{fig:AC_designs}
\end{figure}

Posterior distributions, generated using the PMM, are shown in Figure~\ref{fig:AC_posterior} based on $m_{\mathcal{A}}=5$, $10$, $20$, $40$ and $80$. Posteriors generated with FEM, which do not provide probabilistic quantification for discretisation error, are also included.
Results showed that both the PMM and the FEM method posses similar bias for smaller numbers of design points or coarser meshes, respectively. However while the posteriors generated using FEM are sharply peaked around incorrect values when a coarse mesh is used, a larger, more appropriate variance is reported by the PMM.

\begin{figure}
	\begin{subfigure}{0.49\textwidth}
		\centering
		\includegraphics[width=\textwidth]{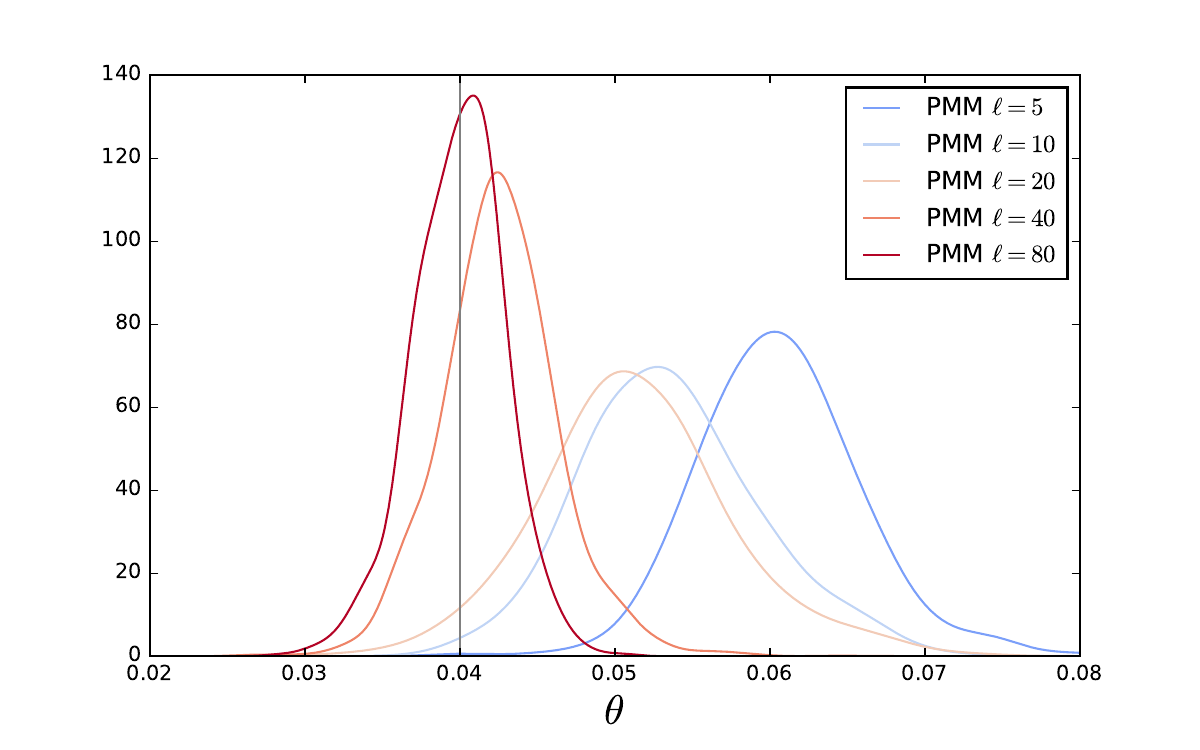}
		\caption{Probabilistic.}
		\label{fig:AC_posterior_meshless}
	\end{subfigure}
	~
	\begin{subfigure}{0.49\textwidth}
		\centering
		\includegraphics[width=\textwidth]{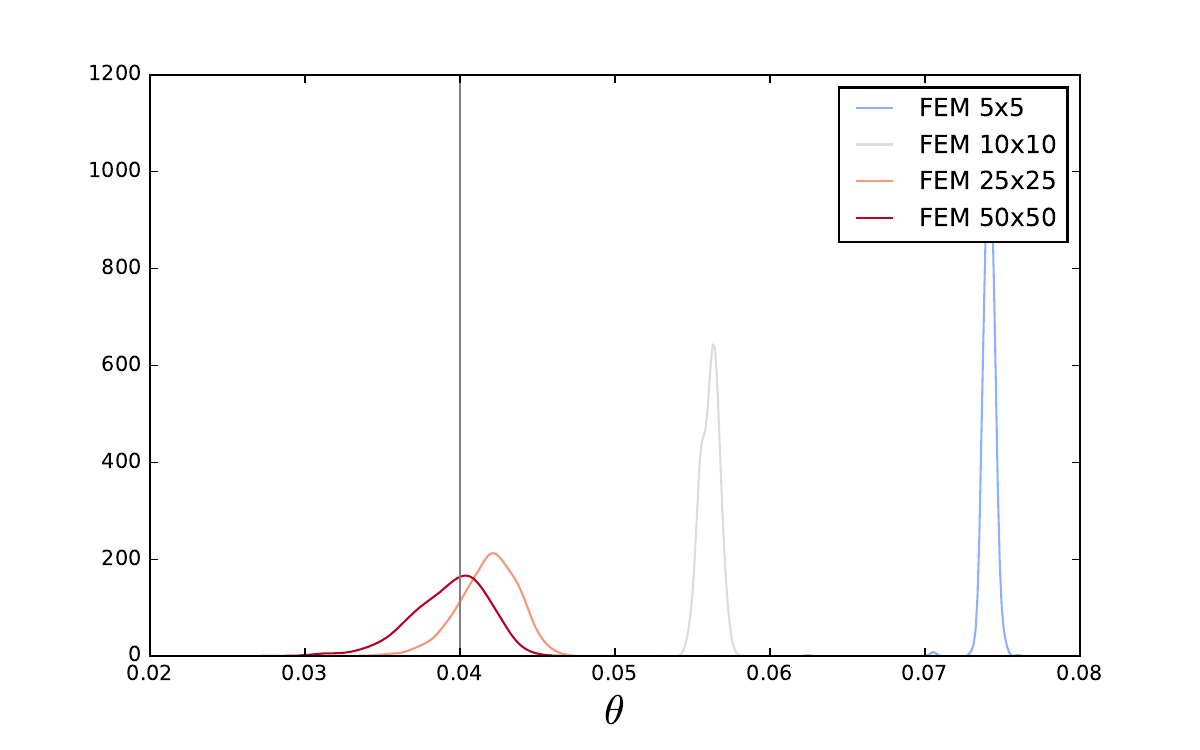}
		\caption{Standard.} \label{fig:AC_posterior_FEA}
	\end{subfigure}
	\caption{Application to Allen--Cahn: Posterior distributions for $\theta$, computed with (a) a probabilistic meshless method (PMM), (b) a finite element method (FEM).
	In (a) the legend denotes the number $m_{\mathcal{A}}$ of design points used by the PMM, while in (b) the legend denotes the size of the mesh used by the FEM.}
	\label{fig:AC_posterior}
\end{figure}

\section{Discussion} \label{sec:discussion}

This paper introduced and explored the concept of a probabilistic meshless method.
While standard numerical solvers return a single approximation to the solution, a probabilistic solver instead returns a full distribution over the solution space, with a view to capturing epistemic uncertainty due to discretisation error.
Our contribution provides theoretical support to this approach and demonstrates that the quantification of numerical error provided by these methods is meaningful, ensuring posterior contraction in an appropriate limit.
Through two example applications, the integration of a probabilistic model for discretisation error into the inverse problem was shown to enable valid statistical inferences to be drawn, overcoming the problems of bias and/or over-confidence that result in practice with standard methods.

Attention was restricted to strong-form solution of stationary PDEs. In future we seek to relax this restriction to examine parabolic, time-evolving PDEs.
To restrict the presentation we did not consider combining probabilistic meshless methods with emulation of the data-likelihood \citep{Stuart2016} or reduced-model approaches \citep{Cohen2015}.
These would provide an obvious and immediate reduction in computational cost in the examples studied in this paper.

The method proposed here shares similarities with the meshless construction recently presented in \cite{Owhadi2016}.
In that work, Owhadi shows how the meshless approach can be made to operate locally in space using a construction called ``gamblets''.
A fundamental distinction of the present paper is that uncertainty due to numerical error was propagated into the inverse problem, whereas \cite{Owhadi2016} does not make use of the probability model beyond observing its Bayesian interpretation.
A second, less fundamental, distinction is that, although Owhadi's method is in principle meshless, computations are performed on a grid, so that there is no analogue of the experimental design approach pursued in this paper.
Future work will aim to combine and leverage the strengths of both methods.

A limitation of this work was the assumption of linearity in the solution $u$ and the restriction to Gaussian prior distributions. 
The first of these assumptions is addressed for a limited class of problems by the work in Section~\ref{sec:nonlinear}. The Gaussian assumption is useful in that posterior distributions can be obtained analytically, but has several drawbacks. In particular, it can be the case that prior information for PDEs is difficult to encode into a Gaussian prior, such as when the solution is known to be bounded. These limitations will be addressed in future work, c.f.\ \cite{Cockayne:2017}.

\appendix
\section{Proofs}

\begin{proof}[Proof of Proposition~\ref{RKHS}]
First, we claim that $\Lambda$ satisfies
\begin{equation}
\Lambda(\bm{x}',\bm{x}) = \int_D G_\Lambda(\bm{x}',\bm{z}) G_\Lambda(\bm{x},\bm{z}) \wrt\bm{z} \label{eq:lambda_claim}
\end{equation}
where $G_\Lambda$ is the Green's function for the system $(\mathcal{A}_\Lambda , \mathcal{B}_\Lambda)$ defined in Eq.~\eqref{eq:laplacian}.
Indeed, since the reproducing kernel must be unique, it suffices to verify that the right-hand side of Eq.~\eqref{eq:lambda_claim} is reproducing in $H_\Lambda(D)$:
\begin{align*}
\left\langle g , \int_D G_\Lambda(\quark,\bm{z}) G_\Lambda(\bm{x},\bm{z}) \wrt\bm{z} \right\rangle_\Lambda
& = \left\langle \mathcal{A}_\Lambda g , \mathcal{A}_\Lambda  \int_D G_\Lambda(\quark,\bm{z}) G_\Lambda(\bm{x},\bm{z}) \wrt\bm{z} \right\rangle_{L^2(D)} \\
& = \left\langle \mathcal{A}_\Lambda g ,  G_\Lambda(\bm{x},\quark) \right\rangle_{L^2(D)} \\
& = g(\bm{x}) ,
\end{align*}
where we have used the definition of the Green's function $G_\Lambda$ and the fact that
\begin{align*}
\mathcal{A}_\Lambda \int_D G_\Lambda(\quark,\bm{z}) G_\Lambda(\bm{x},\bm{z}) \wrt\bm{z}
& = \int_D [\mathcal{A}_\Lambda G_\Lambda(\quark,\bm{z})] G_\Lambda(\bm{x},\bm{z}) \wrt\bm{z} \\
& = \int_D [\delta(\quark - \bm{z})] G_\Lambda(\bm{x},\bm{z}) \wrt\bm{z} \\
& = G_\Lambda(\bm{x},\quark).
\end{align*}
Second, for $v \in H_\textup{nat}(D)$ and using the definition of the Green's functions $G$ and $G_\Lambda$,
\begin{align*}
	\langle v , k_\textup{nat}(\quark,\bm{x}) \rangle_\textup{nat} & = \int_D [\mathcal{A}_\Lambda \mathcal{A} v(\bm{z})] [\mathcal{A}_\Lambda \mathcal{A} k_\textup{nat}(\bm{z},\bm{x})]  \wrt\bm{z} \\
\text{Eq.~\eqref{eq:k_definition}} \implies	& = \int_D\int_D\int_D [\mathcal{A}_\Lambda \mathcal{A} v(\bm{z})] [\mathcal{A}_\Lambda \underbrace{ \mathcal{A} G(\bm{z},\bm{z}') }_{ \delta(\bm{z} - \bm{z}') } G(\bm{x},\bm{z}'') \Lambda(\bm{z}',\bm{z}'') ] \wrt\bm{z} \wrt\bm{z}' \wrt\bm{z}'' \\
	& = \int_D\int_D [\mathcal{A}_\Lambda \mathcal{A} v(\bm{z})] [\underbrace{ \mathcal{A}_\Lambda \Lambda(\bm{z},\bm{z}'') }_{ G_\Lambda(\bm{z},\bm{z}'') } G(\bm{x},\bm{z}'') ] \wrt\bm{z}  \wrt\bm{z}'' \\
	& = \int_D \underbrace{ \int_D [\mathcal{A}_\Lambda \mathcal{A} v(\bm{z})] G_\Lambda(\bm{z},\bm{z}'') \wrt\bm{z} }_{ \mathcal{A} v(\bm{z}'') } G(\bm{x},\bm{z}'')    \wrt\bm{z}''  \\
	& = v(\bm{x}).
\end{align*}
This proves the reproducing property in $H_\textup{nat}(D)$.
By the Cauchy--Schwarz inequality,
\begin{align*}
	|v(\bm{x})| & = |\langle v , k_\textup{nat}(\quark,\bm{x}) \rangle_\textup{nat}| \\
	& \leq \langle v,v \rangle_\textup{nat}^{1/2} \langle k_\textup{nat}(\quark,\bm{x}) , k_\textup{nat}(\quark,\bm{x}) \rangle_\textup{nat}^{1/2} = \|v \|_\textup{nat} k_\textup{nat}(\bm{x},\bm{x})^{1/2},
\end{align*}
which proves the boundedness of the evaluation function.
\end{proof}

\begin{proof}[Proof of Proposition~\ref{noise model}]
	This is essentially Proposition~3.1 of \cite{Owhadi2015}.
	Noting that 
	\begin{equation*}
	u(\bm{x}) = \int_D G(\bm{x},\bm{z}) g(\bm{z}) \wrt\bm{z}, 
	\end{equation*}
	we have
	\begin{align*}
		\mathbb{E}[u(\bm{x}) u(\bm{x}')] &= \mathbb{E}\left[\int_D\int_D G(\bm{x},\bm{z}) g(\bm{z}) G(\bm{x}',\bm{z}') g(\bm{z}') \wrt\bm{z} \wrt\bm{z}'\right] \\
		&= \int_D\int_D G(\bm{x},\bm{z}) G(\bm{x}',\bm{z}') \mathbb{E}[ g(\bm{z}) g(\bm{z}') ] \wrt\bm{z} \wrt\bm{z}' \\
		&= \int_D\int_D G(\bm{x},\bm{z}) G(\bm{x}',\bm{z}') \Lambda(\bm{z},\bm{z}') \wrt\bm{z} \wrt\bm{z}' = k_\textup{nat}(\bm{x},\bm{x}').
	\end{align*}
	Moreover, the stochastic process is well-defined since, from Proposition~\ref{RKHS}, the covariance function $k_\textup{nat}$ is a positive definite function.
\end{proof}

\begin{proof}[Proof of Proposition~\ref{prop:integral_type_kernel}]
Since $D$ is compact and $\tilde{k}$ is symmetric and positive definite, Mercer's theorem \citep{Steinwart2012} guarantees the existence of a countable set of eigenvalues and eigenvectors $\{\lambda_i\}$ and $\{e_i\}$ such that $\lambda_1 \geq \lambda_2 \geq \dots > 0$, $\sum_{i} \lambda_i < \infty$, $\{e_i\}$ are an orthonormal basis of $L^2(D)$ and 
\begin{equation*}
\tilde{k}(\bm{x},\bm{x}') = \sum_{i} \lambda_i e_i(\bm{x}) e_i(\bm{x}').
\end{equation*}
Moreover,
\begin{equation*}
\left\|\sum_{i} c_i \sqrt{\lambda_i} e_i\right\|_{\tilde{k}}^2 = \sum_{i} c_i^2.
\end{equation*}
Then the integral-type kernel $\hat{k}$ can be checked to have eigenvalues and eigenvectors $\{\lambda_i^2\}$ and $\{e_i\}$.
To see that (i) is satisfied, define a stochastic process $S = \sum_i \xi_i \lambda_i e_i$ with $\xi_i \sim N(0,1)$ independent, corresponding to a generic sample from $\Pi_u$.
Then $S$ almost surely lies in $H(D)$, since 
\begin{equation*}
\mathbb{E}\left(\|S\|_{\tilde{k}}^2\right) = \mathbb{E}\left(\sum_{i} \xi_i^2 \lambda_i\right) = \sum_{i} \lambda_i < \infty.
\end{equation*}
To see that (ii) is satisfied, given an element $c = \sum_{i} c_i \sqrt{\lambda_i} e_i$ of $H(D)$ we have $\sum_{i} c_i^2 < \infty$ and so the partial sums $c^{(N)} = \sum_{i=1}^N c_i \sqrt{\lambda_i} e_i$ converge to $c$ under the norm $\|\quark\|_{\tilde{k}}$.
Since each $c^{(N)}$ also belongs to $H_{\hat{k}}(D)$, it follows that the set $H_{\hat{k}}(D)$ is dense in the space $(H(D),\|\quark\|_{\tilde{k}})$.
\end{proof}

\begin{proof}[Proof of Proposition~\ref{prop:natural_local_accuracy}]
First note that 
\begin{equation*}
\mu_\omega(\bm{x}) = \sum_{\mathcal{I} \in \{\mathcal{A} , \mathcal{B}\}} \sum_{j=1}^{n_{\mathcal{I}}} w_j^{\mathcal{I}} \mathcal{I}u(\bm{x}_{0,j}^{\mathcal{I}})
\end{equation*}
where the weights are 
\begin{equation}
	\label{proof weights}
	\begin{bmatrix} \bm{w}^{\mathcal{A}} \\ \bm{w}^{\mathcal{B}} \end{bmatrix}^\top  = \bar{\mathcal{L}} \hat{\bm{K}}(\bm{x},X_0) [\mathcal{L}\bar{\mathcal{L}} \hat{\bm{K}}(X_0)]^{-1}.
\end{equation}
	Now, from the reproducing property, we have
	\begin{equation*}
		\mu(\bm{x}) = \sum_{\mathcal{I} \in \{\mathcal{A} , \mathcal{B}\}} \sum_{j=1}^{n_{\mathcal{I}}} w_j^{\mathcal{I}} \mathcal{I} \inner{ u_\omega , \hat{k}(\quark,\bm{x}_{0,j}^{\mathcal{I}}) }_{\hat{k}}
			= \inner{ u_\omega , \sum_{\mathcal{I} \in \{\mathcal{A} , \mathcal{B}\}} \sum_{j=1}^{n_{\mathcal{I}}} w_j^{\mathcal{I}} \bar{\mathcal{I}} \hat{k}(\quark,\bm{x}_{0,j}^{\mathcal{I}}) }_{\hat{k}}
	\end{equation*}
	and hence, using the reproducing property again,
	\begin{equation*}
		u_\omega(\bm{x}) - \mu_\omega(\bm{x}) = \inner{u , \hat{k}(\quark,\bm{x}) -  \sum_{\mathcal{I} \in \{\mathcal{A} , \mathcal{B}\}} \sum_{j=1}^{n_{\mathcal{I}}} w_j^{\mathcal{I}} \bar{\mathcal{I}} \hat{k}(\quark,\bm{x}_{0,j}^{\mathcal{I}}) }_{\hat{k}}.
	\end{equation*}
	Finally, the Cauchy--Schwarz inequality yields
	\[
		\abs{u_\omega(\bm{x}) - \mu_\omega(\bm{x})} \leq \norm{u_\omega }_{\hat{k}} \norm{ \hat{k}(\quark,\bm{x}) - \sum_{\mathcal{I} \in \{\mathcal{A} , \mathcal{B}\}} \sum_{j=1}^{n_{\mathcal{I}}} w_j^{\mathcal{I}} \bar{\mathcal{I}} \hat{k}(\quark,\bm{x}_{0,j}^{\mathcal{I}}) }_{\hat{k}}.
	\]
	Upon substitution of the expression for the weights $w_j^{\mathcal{I}}$ provided in Eq.~\eqref{proof weights}, the second term is recognised as $\sigma_\omega(\bm{x})$.
\end{proof}

\begin{proof}[Proof of Theorem~\ref{thm:fwd_contraction}]
Suppose $u$ is a random variable with distribution $\Pi_u^{\bm{g},\bm{b}}$.
Then we have
\begin{align*}
	\int_{\pnspace} \|u - u_\omega\|_2^2 \wrt \Pi_u^{\bm{g},\bm{b}}
	& \leq \int_{\pnspace} \|u - \mu_\omega\|_2^2 \wrt \Pi_u^{\bm{g},\bm{b}} + \underbrace{\int_{\pnspace} \|\mu_\omega - u_\omega\|_2^2 \wrt \Pi_u^{\bm{g},\bm{b}}}_{\text{indep.\ of }u} \\
	& = \int_D \int_{\pnspace} (u(\bm{x}) - \mu_\omega(\bm{x}))^2 \wrt \Pi_u^{\bm{g},\bm{b}} \wrt\bm{x} + \int_D (\mu_\omega(\bm{x}) - u_\omega(\bm{x}))^2 \wrt \bm{x} \\
	& \leq \int_D \sigma_\omega(\bm{x})^2 \wrt\bm{x} + \|u_\omega\|_{\hat{k}}^2 \int_D \sigma_\omega(\bm{x})^2 \wrt\bm{x}
\end{align*}
where the second line uses Fubini's theorem to interchange the order of integration and the final line makes use of Proposition~\ref{prop:natural_local_accuracy}.
Since the domain $D$ is bounded, we have from Proposition~\ref{prop:sigma} that there exists a constant $C_{\omega}^F$, dependent on $\omega$ but independent of $\pnparam$, for which $\int_D \sigma_\omega(\bm{x})^2 \wrt\bm{x} \leq C_{\omega}^F h^{2\beta - 2\rho - d}$ and therefore
\begin{equation*}
\int_{\pnspace} \|u - u_\omega\|_2^2 \wrt \Pi_u^{\bm{g},\bm{b}} \leq C_{\omega}^F (1 + \|u_\omega\|_{\hat{k}}^2) h^{2\beta - 2\rho - d}.
\end{equation*}
The result follows from Markov's inequality:
for any $\epsilon > 0$,
\begin{equation*}
\Pi_u^{\bm{g},\bm{b}}\{u \in H(D) : \|u - u_\omega\|_2^2 > \epsilon \} \leq \frac{\int \|u - u_\omega\|_2^2 \; \wrt \Pi_u^{\bm{g},\bm{b}}}{\epsilon} \leq \frac{C_{\omega}^F (1 + \|u_\omega\|_{\hat{k}}^2) h^{2\beta - 2\rho - d}}{\epsilon}
\end{equation*}
as required.
\end{proof}

\begin{proof}[Proof of Theorem~\ref{thm:well_defined}]
The proof below consists of three steps.

\medskip
\noindent \textbf{Step \#1:}
Fix $\omega \in \Omega$ and $t \leq 0$.
Consider an element $u$ in the Hilbert scale of spaces $H_\textup{nat}^t(D)$.
Using the fact that $\sqrt{\lambda_i} e_i$ are an orthonormal basis for $H_\textup{nat}^t(D)$, we have
\begin{equation*}
\|u\|_{\textup{nat},t}^2 = \sum_{i=1}^\infty \lambda_i^{-t} \langle u , \sqrt{\lambda_i} e_i \rangle_\textup{nat}^2.
\end{equation*}
By construction, a generic element $g \in H_\Lambda(D)$ can be written as $g = \sum_{i=1}^\infty c_i g_i$ where $g_i = \mathcal{A} \sqrt{\lambda_i} e_i$ form an orthonormal basis $\{g_i\}_{i=1}^\infty$ for $H_\Lambda(D)$ and $\|g\|_\Lambda^2 = \sum_{i=1}^\infty c_i^2 < \infty$.
Thus we have
\begin{equation*}
\|u\|_{\textup{nat},t}^2 = \sum_{i=1}^\infty \lambda_i^{-t} \langle \mathcal{A} u , \mathcal{A} \sqrt{\lambda_i} e_i \rangle_\Lambda^2 = \sum_{i=1}^\infty \lambda_i^{-t} \langle \mathcal{A} u , g_i \rangle_\Lambda^2 = \|\mathcal{A} u\|_{\Lambda,t}^2.
\end{equation*}

\medskip
\noindent \textbf{Step \#2:}
The stochastic process $g(\quark,\pnparam)$ with kernel $\Lambda$ can be characterised through the Karhunen--Lo\`{e}ve expansion as $g(\quark,\pnparam) = \sum_{i=1}^\infty \xi_i(\pnparam) g_i(\quark)$, where the $\xi_i(\pnparam)$ are independent standard normal random variables under $\mathbb{P}_{\pnspace}$ and the $g_i(\quark) = \mathcal{A} \sqrt{\lambda_i} e_i(\quark)$ were defined in Step \#1.
From sub-additivity of measure, we have that 
\begin{align*}
\mathbb{E}_{\pnspace} \|g(\quark,\pnparam)\|_{\Lambda,t}^2
= \mathbb{E}_{\pnspace} \sum_{i=1}^\infty \lambda_i^{-t} \xi_i(\pnparam)^2 
\leq \sum_{i=1}^\infty \lambda_i^{-t} \mathbb{E}_{\pnspace} \xi_i(\pnparam)^2 
= \sum_{i=1}^\infty \lambda_i^{-t}.
\end{align*}
From (A2) we obtain (recall $t < -d/2\alpha < 0$)
\begin{equation*}
\mathbb{E}_{\pnspace} \|g(\quark,\pnparam)\|_{\Lambda,t}^2 \; \leq \; C_\omega^{t} \sum_{i=1}^\infty [\lambda_i^{(\alpha)}]^{-t}.
\end{equation*}
Combining this with the result of Step \#1 implies that
\begin{equation*}
\mathbb{E}_{\pnspace} \|u(\quark,\omega,\pnparam)\|_{\textup{nat},t}^2 = \mathbb{E}_{\pnspace} \|\mathcal{A}u(\quark,\omega,\pnparam)\|_{\Lambda,t}^2 = \mathbb{E}_{\pnspace} \|g(\quark,\pnparam)\|_{\Lambda,t}^2 \leq C_\omega^{t} \sum_{i=1}^\infty [\lambda_i^{(\alpha)}]^{-t}.
\end{equation*}

\medskip
\noindent \textbf{Step \#3:}
Consider the doubly stochastic process $u(\quark,\omega,\pnparam)$ and a double expectation $\mathbb{E}_\Omega \mathbb{E}_{\pnspace}$ over $\omega$ and $\pnparam$.
From (A2) we have
\[
\mathbb{E}_\Omega \mathbb{E}_{\pnspace} \|u(\quark,\omega,\pnparam)\|_{\mathbb{H}^\alpha(D),t}^2 \leq \mathbb{E}_\Omega \; C_{\omega,t} \; \mathbb{E}_{\pnspace} \|u(\quark,\omega,\pnparam)\|_{\textup{nat},t}^2.
\]
The output of Step \#2 then implies that
\[
\mathbb{E}_\Omega \mathbb{E}_{\pnspace} \|u(\quark,\omega,\pnparam)\|_{\mathbb{H}^\alpha(D),t}^2 \; \leq \;  \underbrace{\mathbb{E}_\Omega[\; C_\omega^{t} \; C_{\omega,t} \; ]}_{(\ast)} \sum_{i=1}^\infty [\lambda_i^{(\alpha)}]^{-t}.
\]
Under (A2), the term $(\ast)$ is finite when $t < -d/2\alpha$.
Observe that, since $\lambda_i^{(\alpha)} \asymp i^{-2\alpha / d}$, the right hand side is finite for all values of $t < - d / 2\alpha$.
This implies that $u \in L_{\mathbb{P}_\Omega \mathbb{P}_{\pnspace}}^2(\Omega , \pnspace; [\mathbb{H}^{\alpha}(D)]^t)$ for all $-1 < t < -d / 2 \alpha$ and hence $u \in L_{\mathbb{P}_\Omega \mathbb{P}_{\pnspace}}^2(\Omega , \pnspace; \mathbb{H}^s(D))$ for all $0 < s < \alpha - d/2$.
\end{proof}


\begin{proof}[Proof of Theorem~\ref{thm:theta_consistency}]

Throughout we use the L\"{o}wner ordering $\preceq$ on positive semidefinite matrices, i.e.\ $\bm{A} \preceq \bm{B}$ if and only if $\bm{B} - \bm{A}$ is positive semidefinite. 
In particular, we use the facts that $\bm{x}^\top \bm{A} \bm{x} \leq \bm{x}^\top \bm{B} \bm{x}$ and $\bm{B}^{-1} \preceq \bm{A}^{-1}$ whenever $\bm{A} \preceq \bm{B}$.
For more information, see \cite{Bernstein:2009}.

From \cite[Theorem~4.9]{Dashti2014}, it is sufficient to show that the two potentials $\Phi_h(\bm{y},\theta)$ and $\Phi(\bm{y},\theta)$ are asymptotically identical. Since $\bm{y}$ is fixed throughout we suppress this argument, and consider $\abs{\Phi_h(\theta) - \Phi(\theta)}$. Let $\Phi_{\text{coll}}(\theta) = (\bm{y} - \bm{\mu})^\top \bm{\Gamma}^{-1} (\bm{y} - \bm{\mu})$ denote the approximate potential based on symmetric collocation and with discretisation error ignored. 
Then 
\begin{equation}
	\abs{ \Phi_h(\theta) - \Phi(\theta) } \leq \abs{\Phi_h(\theta) - \Phi_\text{coll}(\theta)} + \abs{\Phi_\text{coll}(\theta) - \Phi(\theta)}  . \label{eq:theta_proof_main_ineq}
\end{equation}

Considering the second of these terms, we have
\begin{align*}
	& \abs{\Phi_\text{coll}(\theta) - \Phi(\theta)} \\
	&= \abs{
		(\bm{y} - \bm{\mu})^\top \bm{\Gamma}^{-1} (\bm{y} - \bm{\mu}) - \Phi(\theta)
	} \\
	&= \abs{
		(\bm{y} - \bm{u})^\top \bm{\Gamma}^{-1} (\bm{y} - \bm{u}) 
		+ 2(\bm{y} - \bm{u})^\top \bm{\Gamma}^{-1}(\bm{u} - \bm{\mu}) 
		+ (\bm{u} - \bm{\mu})^\top \bm{\Gamma}^{-1} (\bm{u} - \bm{\mu}) - \Phi(\theta)
	} \\
	&= \abs{
		2(\bm{y} - \bm{u})^\top \bm{\Gamma}^{-1}(\bm{u} - \bm{\mu}) 
		+ (\bm{u} - \bm{\mu})^\top \bm{\Gamma}^{-1} (\bm{u} - \bm{\mu})
	} \\
	&\leq \underbrace{(\bm{u} - \bm{\mu})^\top \bm{\Gamma}^{-1}(\bm{u} - \bm{\mu})}_{(a)}
	+ 2\underbrace{\abs{(\bm{y} - \bm{u})^\top \bm{\Gamma}^{-1} (\bm{u} - \bm{\mu})}}_{(b)} .
\end{align*}

Now each of terms $(a)$ and $(b)$ can be bounded since, for a positive semidefinite matrix $\bm{A}$ with maximal eigenvalue $\lambda_{\max}[\bm{A}]$ it holds that $\bm{A} \preceq \lambda_{\max}[\bm{A}] \bm{I}$. 
Let $\norm{\bm{x}}_{\bm{\Gamma}} := \sqrt{\bm{x}^\top \bm{\Gamma}^{-1} \bm{x}}$ and let $\gamma = \left(\lambda_{\min}[\bm{\Gamma}]\right)^{-1}$. 
Then $\|\bm{x}\|_{\bm{\Gamma}} \leq \sqrt{\gamma} \|\bm{x}\|_2$.
Thus for $(a)$:
\begin{align*}
	\abs{(\bm{u} - \bm{\mu})^\top \bm{\Gamma}^{-1} (\bm{u} - \bm{\mu})}
		&\leq \gamma \norm{\bm{u} - \bm{\mu}}_2^2 \\
		&\leq \gamma\; n \norm{\bm{u} - \bm{\mu}}_\infty^2 .
\end{align*}
Similarly for $(b)$:
\begin{align*}
	\abs{(\bm{y} - \bm{u})^\top \bm{\Gamma}^{-1} (\bm{u} - \bm{\mu})} &\leq 
		\left[
			(\bm{y} - \bm{u})^\top \bm{\Gamma}^{-1} (\bm{y} - \bm{u}) \cdot
			(\bm{u} - \bm{\mu})^\top \bm{\Gamma}^{-1} (\bm{u} - \bm{\mu}) 
		\right]^{\frac{1}{2}} \\
		&\leq \gamma \norm{\bm{y} - \bm{u}}_2 \norm{\bm{u} - \bm{\mu}}_2 \\
		&\leq \gamma \norm{\bm{y} - \bm{u}}_2 \cdot \sqrt{n}\norm{\bm{u} - \bm{\mu}}_\infty
\end{align*}
where the first inequality used is the Cauchy--Schwarz inequality, while the final line uses 
$\norm{\quark}_2 \leq \sqrt{n}\norm{\quark}_\infty$.

Now returning to the first term in Eq.~\eqref{eq:theta_proof_main_ineq}
\begin{align*}
	\abs{\Phi_h(\theta) - \Phi_\text{coll}(\theta)} &= \abs{
		(\bm{y} - \bm{\mu})^\top (\bm{\Sigma} + \bm{\Gamma})^{-1}(\bm{y} - \bm{\mu}) - (\bm{y} - \bm{\mu})^\top \bm{\Gamma}^{-1} (\bm{y} - \bm{\mu})
	} .
\end{align*}
Applying the Woodbury identity we obtain
\begin{align*}
	(\bm{y} - \bm{\mu})^\top (\bm{\Sigma} + \bm{\Gamma})^{-1} (\bm{y} - \bm{\mu}) 
		&= (\bm{y} - \bm{\mu})^\top \bm{\Gamma}^{-1}(\bm{y} - \bm{\mu}) \\
		& \qquad - (\bm{y} - \bm{\mu})^\top \bm{\Gamma}^{-1} (\bm{\Sigma}^{-1} + \bm{\Gamma}^{-1})^{-1} \bm{\Gamma}^{-1} (\bm{y} - \bm{\mu})
\end{align*}
and so, letting $\bm{M} := \bm{\Gamma}^{-1} (\bm{\Sigma}^{-1} + \bm{\Gamma}^{-1})^{-1} \bm{\Gamma}^{-1}$ and applying the triangle inequality in the $\norm{\cdot}_{\bm{M}^{-1}}$ norm we obtain
\begin{align*}
	\abs{\Phi_h(\theta) - \Phi_\text{coll}(\theta)} &= (\bm{y} - \bm{\mu})^\top \bm{M} (\bm{y} - \bm{\mu}) \\
	&\leq \underbrace{ (\bm{y} - \bm{u})^\top \bm{M} (\bm{y} - \bm{u})}_{(c)} 
		+ \underbrace{(\bm{u} - \bm{\mu})^\top \bm{M} (\bm{u} - \bm{\mu})}_{(d)} .
\end{align*}
For $(c)$, we note that $\bm{M} \preceq \bm{\Gamma}^{-1} \bm{\Sigma} \bm{\Gamma}^{-1}$ and arrive at
\begin{align*}
(\bm{y} - \bm{u})^\top \bm{M} (\bm{y} - \bm{u}) & \leq (\bm{y} - \bm{u})^\top \bm{\Gamma}^{-1} \bm{\Sigma} \bm{\Gamma}^{-1} (\bm{y} - \bm{u}) \\
& \leq \gamma^2 \text{Tr}(\bm{\Sigma}) \norm{\bm{y} - \bm{u}}_2^2.
\end{align*}
Now note that $\bm{M} \preceq \bm{\Gamma}^{-1}$, since $(\bm{\Sigma}^{-1} + \bm{\Gamma}^{-1})^{-1} \preceq \bm{\Gamma}$. For term $(d)$, this means
\begin{align*}
	(\bm{u} - \bm{\mu})^\top \bm{M} (\bm{u} - \bm{\mu}) &\leq (\bm{u} - \bm{\mu})^\top \bm{\Gamma}^{-1} (\bm{u} - \bm{\mu})
\end{align*}
which we recognise as $(a)$ above. 

Combining these bounds and using Propositions~\ref{prop:natural_local_accuracy} and \ref{prop:sigma}, we have that
\begin{align*}
	\abs{\Phi_h(\theta) - \Phi(\theta)} &\leq 2 \gamma \sqrt{n}  \norm{\bm{y}-\bm{u}}_2 \norm{\bm{u} - \bm{\mu}}_\infty
		+ 2\gamma n \norm{\bm{u} - \bm{\mu}}_\infty^2 
		+ \gamma^2 \trace(\bm{\Sigma}) \norm{\bm{y} - \bm{u}}_2^2 \\
	&\leq 
		2 \gamma \sqrt{n} \norm{\bm{y} - \bm{u}}_2 \norm{u_\omega}_{\hat{k}} C_\omega^F h^{\beta - \rho - d/2} \\
	&\qquad+ \left[
		2\gamma n \norm{u_\omega}_{\hat{k}}^2
		+ \gamma^2 n \norm{\bm{y} - \bm{u}}_2^2
	\right] (C_\omega^F)^2 h^{2\beta - 2\rho - d} \\
	&\leq h^{\beta - \rho - d/2} \bigg[ 
		2\gamma\sqrt{n} \norm{\bm{y} - \bm{u}}_2  \norm{u_\omega}_{\hat{k}} C^F_\omega 
		+ 2\gamma n  \norm{u_\omega}_{\hat{k}}^2 (C^F_\omega)^2 \\
	&\qquad\qquad\qquad
		+ \gamma^2 n \norm{\bm{y}-\bm{u}}_2^2(C^F_\omega)^2
	\bigg]
\end{align*}
for $h < 1$.

It remains to show that the assumptions required for \cite[Theorem 4.9]{Dashti2014} hold, namedly that there exist functions $M_1,M_2:\reals^+\to\reals^+$ so that
\begin{enumerate}
	\item[D1:] $\Phi(\theta) \geq -M_1(\norm{\theta}_\Theta)$,
	\item[D2:] $\Phi_h(\theta) \geq -M_1(\norm{\theta}_\Theta)$,
	\item[D3:] $\abs{\Phi_h(\theta) - \Phi(\theta)} \leq M_2(\norm{\theta}_\Theta) \varphi(h)$,
	\item[D4:] $\exp(M_1(\norm{\theta}_\Theta) \left( 1 + M_2(\norm{\theta}_\Theta)^2\right)$ is integrable in $\theta$ with respect to $\mu$. \label{assumption:dashti_4}
\end{enumerate}
where $\varphi(h) \to 0$ as $h \to 0$.
If (D1-4) hold then it can be concluded that, for $h$ sufficiently small, $d_{\text{Hell}}(\Pi_\theta^{\bm{y},h},\Pi_\theta^{\bm{y}}) \leq C \varphi(h)$ for some constant $C$, as required.
Clearly we can take $M_1(\norm{\theta}_\Theta) = 0$ to satisfy both (D1) and (D2), as $\Phi(\theta) \geq 0$ and $\Phi_h(\theta) \geq 0$ for all $\theta$. 
We also take $\varphi(h) = h^{\beta - \rho - d/2}$.
Let $\eta = \sup\{ \norm{\bm{u}}_2 : \norm{u}_{\hat{k}} \leq 1 \}$ ($< \infty$).
To define $M_2(\norm{\theta}_\Theta)$, first note:
\begin{align*}
	2\gamma\sqrt{n} \norm{\bm{y} - \bm{u}}_2 \norm{u_\omega}_{\hat{k}} C^F_\omega 
		& \leq 2 \gamma \sqrt{n}  (\norm{\bm{y}}_2 + \norm{\bm{u}}_2)\norm{u_\omega}_{\hat{k}} C_\omega^F \\
		&\leq 2\gamma\sqrt{n} (\norm{\bm{y}}_2 + \eta\norm{u_\omega}_{\hat{k}}) \norm{u_\omega}_{\hat{k}}C_\omega^F \\
		&\leq 2\gamma\sqrt{n} (\norm{\bm{y}}_2 + \eta) C(\norm{\theta}_\Theta) \\
\end{align*}
and also
\begin{align*}
	& [2 \gamma n \norm{u_\omega}_{\hat{k}}^2 + \gamma^2 n \norm{\bm{y} - \bm{u}}_2^2 ] (C_\omega^F)^2 \quad  \\
	&\quad \leq \left[ 
		2 \gamma n \norm{u_\omega}_{\hat{k}}^2 
		+ \gamma^2 n \norm{\bm{y}}_2^2 
		+ 2 \gamma^2 n \norm{\bm{y}}_2 \norm{\bm{u}}_2
		+ \gamma^2 n \norm{\bm{u}}_2^2 
	\right] (C_\omega^F)^2  \\
	& \quad \leq  \left[ 
		2 \gamma n \norm{u_\omega}_{\hat{k}}^2 
		+ \gamma^2 n \norm{\bm{y}}_2^2 
		+ 2 \gamma^2 n  \norm{\bm{y}}_2 \eta \norm{u_\omega}_{\hat{k}}
		+ \gamma^2 n \eta^2 \norm{u_\omega}_{\hat{k}}^2
	\right] (C_\omega^F)^2  \\
	& \quad \leq \left[ 
		2 \gamma n 
		+ \gamma^2 n \norm{\bm{y}}_2^2 
		+ 2 \gamma^2 n \eta
		+ \gamma^2 n \eta^2 
	\right] C(\norm{\theta}_\Theta)^2
\end{align*}
where $C(\norm{\theta}_\Theta)$ is as in (A3). 
Lastly, define
\begin{align*}
	M_2(\norm{\theta}_\Theta) &:= 
	2\gamma\sqrt{n} (\norm{\bm{y}}_2 + \eta) C(\norm{\theta}_\Theta) \\
	&\qquad+  \left[ 
		2 \gamma n 
		+ \gamma^2 n \norm{\bm{y}}_2^2 
		+ 2 \gamma^2 n \eta
		+ \gamma^2 n \eta^2 
	\right] C(\norm{\theta}_\Theta)^2.
\end{align*}
Then by construction (D3) is satisfied.
Furthermore (D4) is satisfied by the second part of (A3).
This completes the proof.

\end{proof}

\vspace{20pt}

\section*{Acknowledgement} 


The authors thank John Skilling for insight, Patrick Farrell for use of the Matlab source code used in \cite{Farrell2014}, and Fran\c{c}ois-Xavier Briol for feedback.
In addition they are grateful to the developers of the Python libraries \href{https://github.com/HIPS/autograd}{Autograd} and \href{http://sheffieldml.github.io/GPyOpt/}{GPyOpt}.



\bibliographystyle{plainnat}
\bibliography{jabref}

\end{document}